\documentclass[conference,compsoc]{IEEEtran}

\ifCLASSOPTIONcompsoc
  \usepackage[nocompress]{cite}
\else
  \usepackage{cite}
\fi

\usepackage{siiscmds}                           
\usepackage[ruled,linesnumbered]{algorithm2e}    
\usepackage{algpseudocode}
\SetCommentSty{itshape}
\usepackage{booktabs}                            
\usepackage{centernot}                           
\usepackage{colortbl}                            
\usepackage{graphicx}                            
\usepackage{hhline}                              
\usepackage{lipsum}                              
\usepackage{mathtools}                           
\usepackage[all]{nowidow}                        
\usepackage{optidef}                             
\usepackage{pgfplotstable}                       
\usepackage{subcaption}                          
\usepackage[binary-units]{siunitx}                             
\usepackage{substr}                              
\usepackage{tikz}                                
\usepackage{tcolorbox}                           
\usepackage{xcolor}                              
\usepackage{circledsteps}
\usepackage[breaklinks]{hyperref}
\definecolor{light-gray}{gray}{0.968}

\usepackage{xspace}                              
\usepackage{glossaries}
\usepackage{multirow}
\usepackage{float}
\usepackage{url}
\usepackage{pgf}
\usepackage{wrapfig}
\usepackage{pifont}
\usepackage{epsfig}
\usepackage{endnotes}
\usepackage{textcomp}
\usepackage{comment}
\usepackage{enumitem}
\usepackage{microtype}
\usepackage{tabularx}
\usepackage{amsthm}
\usepackage[para]{threeparttable}

\usepackage{tikz}
\newcommand*\emptycirc[1][1ex]{\tikz\draw[thick] (0,0) circle (#1);} 
\newcommand*\halfcirc[1][1ex]{%
  \begin{tikzpicture}
  \draw[fill] (0,0)-- (90:#1) arc (90:270:#1) -- cycle ;
  \draw[thick] (0,0) circle (#1);
  \end{tikzpicture}}
\newcommand*\fullcirc[1][1ex]{%
  \begin{tikzpicture}
  \draw[fill] (0,0)-- (0:#1) arc (0:360:#1) -- cycle ;
  \draw[thick] (0,0) circle (#1);
    \end{tikzpicture}}

\newcommand{\tool}{\texttt{PAC}}
\newcommand{\FSYNC}{\texttt{FSYNC}}

\newtheorem{theorem}{Theorem}

\definecolor{OliveGreen}{rgb}{0,0.6,0}

\newlist{rqlist}{enumerate}{3}
\setlist[rqlist]{label=\textbf{RQ\arabic*}.,before=\raggedright,leftmargin=40pt,ref=RQ\arabic*}

\newcommand{\newtxt}[1]{#1}
\newcommand{\oldtxt}[1]{}

\newcommand{\itembase}[1]{\setlength{\itemsep}{#1}}

\begin{document}

\title{Efficient Storage Integrity in Adversarial Settings}



%
\IEEEoverridecommandlockouts 
\author{
\IEEEauthorblockN{
Quinn Burke\IEEEauthorrefmark{1}\thanks{\IEEEauthorrefmark{1}Corresponding author (email: qkb@cs.wisc.edu)}, Ryan Sheatsley, Yohan Beugin, Eric Pauley,\\Owen Hines, Michael Swift, Patrick McDaniel}
\IEEEauthorblockA{University of Wisconsin-Madison}}


\IEEEoverridecommandlockouts
\makeatletter
\def\ps@customtitlepagestyle{%
  \def\@oddhead{%
    \parbox{\textwidth}{%
      \centering
      \makebox[\textwidth]{%
        \hfill\rule{0.15\textwidth}{0.4pt}\hspace{1em}%
        2025 IEEE Symposium on Security and Privacy (S\&P)%
        \hspace{1em}\rule{0.15\textwidth}{0.4pt}\hfill
      }%
    }%
  }%
  \def\@oddfoot{}%
}
\makeatother

\maketitle
\thispagestyle{customtitlepagestyle} 
\pagestyle{plain} 


\begin{abstract}
Storage integrity is essential to systems and applications that use untrusted storage (e.g., public clouds, end-user devices). However, known methods for achieving storage integrity either suffer from high (and often prohibitive) overheads or provide weak integrity guarantees.
In this work, we demonstrate a hybrid approach to storage integrity that simultaneously reduces overhead while providing strong integrity guarantees. Our system, partially asynchronous integrity checking (\tool{}), allows disk write commitments to be deferred while still providing guarantees around read integrity. \tool~delivers a $5.5\times$ throughput and latency improvement over the state of the art, and $85\%$ of the throughput achieved by non-integrity-assuring approaches. In this way, we show that untrusted storage can be used for integrity-critical workloads without meaningfully sacrificing performance.

\end{abstract}


%
\IEEEpeerreviewmaketitle
\section{Introduction}
\label{sec:introduction}
Storage integrity is essential to systems and applications that use untrusted storage (e.g., public clouds, end-user devices).  
While robust systems for ensuring the integrity of the compute environment have become commonplace on cloud providers~\cite{aws_sev_snp,gcp_confidential_vm,azure_confidential_vm} and mobile devices~\cite{android-dm-verity}, the storage~\cite{aws-ebs,gc-pdisk,azure-mdisk} connected to these environments do not guarantee that returned data is authentic (i.e., was actually written by the application), transactionally consistent (i.e., the storage state is globally correct with respect to storage operations), and fresh (i.e., has not been rolled back to a previous state). When these properties are compromised, an adversarial storage device (be it a malicious cloud provider, backdoored physical device, or low-level system malware) can replay financial transactions, undo deployment of vulnerability mitigations, and otherwise compromise the integrity of the system as a whole~\cite{johannesmeyer2024practical,huang_survey_2014,bohling2020subverting,kurmus2017random}.

While local integrity is readily achieved using authenticated encryption (e.g., AES/AEAD), global transactional consistency and freshness on untrusted storage requires the use of globally-consistent data structures. Researchers have converged on Merkle hash trees---which recursively hash the data stored in blocks to one single checksum---to track the globally-consistent state of a storage device~\cite{merkle1989certified,gassend2003caches,mckeen_innovative_2013,khati2017full,priebe_sgx-lkl_2020,android-dm-verity}. When tracked using trusted hardware~\cite{tsai_graphene-sgx_nodate,priebe_enclavedb_2018,perez2006vtpm,taassori2018vault} (e.g., secure counters offered by the CPU), this approach provides transactional consistency and prevents an adversary from rolling the storage system back to a previous state. Such a solution has also been adopted by practitioners, and is available in Linux as dm-verity~\cite{android-dm-verity}.

Unfortunately, using hash trees to achieve global consistency comes at a dramatic cost to performance, particularly as storage sizes scale~\cite{crosby2011authenticated}. Prior works have approached this problem in two ways: (a) by optimizing Merkle tree implementations using caches~\cite{gassend2003caches} and new tree implementations~\cite{taassori2018vault,chakraborti2017dm,burke2025scalable}, and (b) by relaxing the storage threat model to no longer provide strong protections against adversarial data rollbacks and data manipulation~\cite{arasu2021fastver,arasu2017concerto}. Under these relaxed threat models, techniques defer the durable storage of written data \textit{and} the verification of disk reads, under the premise that these operations incur unacceptable performance overhead. However, this means that an application can proceed after reading incorrect data, taking actions that compromise overall system integrity. In effect, the existing space of disk integrity techniques presents a seemingly-fundamental tradeoff between strong security guarantees and performance.

In this work, we examine this security-performance tradeoff. Current techniques exist on two ends of a spectrum: traditional \textit{synchronous} approaches persist and verify all writes and reads in order and as they are performed, and \textit{asynchronous} approaches defer each of these in the interest of performance. We hypothesize that a \textit{hybrid} of asynchronous writes and coordinated synchronous reads can virtually eliminate the overhead of transactional disk integrity, while still providing the same security guarantees as synchronous approaches. Our approach, \underline{p}artially \underline{a}synchronous integrity \underline{c}hecking (\tool), ensures complete transactional integrity of disk contents at runtime, and provides rollback protection of fully-committed writes under adversarial fault injection at the storage layer (the same guarantees offered by conventional disk write semantics in a benign setting).

We evaluate \tool{} in two stages, beginning with a formal security analysis. We establish two guarantees: (1) a \textit{read guarantee}, wherein reads always reflects the most recent acknowledged write at the time data is returned, and (2) a \textit{write guarantee}, which ensures that the sealed and durable Merkle root always reflects the most recent \FSYNC{} call. We provide a set of proofs to establish these guarantees under \tool's hybrid approach.

We next implement \tool{} as a Linux block device driver and evaluate its performance, measuring its overhead, scalability, and memory/storage trade-off. \tool{} delivers $>85\%$ of the performance of block-level AEAD while \textit{additionally} providing global transactional integrity. It scales to large storage devices, with throughputs up to $5.5\times$ that of prior hash tree designs. Finally, \tool{} achieves these with minimal additional memory requirements. In these ways, \tool{} enables the strongest storage integrity guarantees at minimal cost. 

Software systems rely on the integrity of both computation and storage to ensure end-to-end outcomes. Coupled with advances in trusted execution technology~\cite{mckeen_innovative_2013,tsai_graphene-sgx_nodate,priebe_sgx-lkl_2020}, our work shows that the transactional integrity and freshness of underlying storage can be achieved at minimal overhead. We anticipate that \tool{} will enable the practical use of trusted storage in an expanded space of performance- and security-critical software systems. \newtxt{Our code is plug-and-play into standard Linux systems and open-sourced at \href{https://github.com/MadSP-McDaniel/pac}{https://github.com/MadSP-McDaniel/pac}}.

\section{Background}
\label{sec:background}

\begin{figure}[t]
    \centering
    \includegraphics[width=\linewidth]{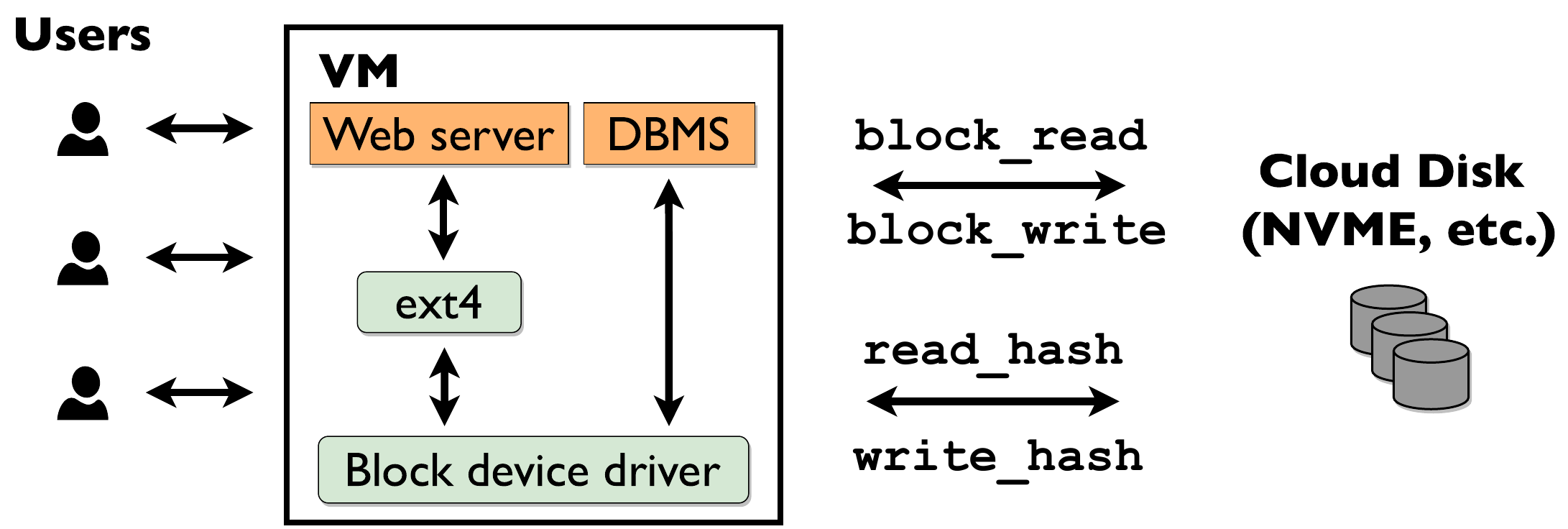}
    \caption{We consider an IaaS deployment model where an application runs inside a guest VM and stores data on a fast, local NVMe disk.}
    \label{fig:system-model}
\end{figure}

\shortsection{Cloud Disks}
Block storage is a backbone of modern public cloud infrastructure~\cite{aws-ebs,gc-pdisk,azure-mdisk}. While there are various deployment models for cloud applications and storage, we consider a standard Infrastructure-as-a-Service (IaaS) deployment where an application runs inside a guest VM and reads and writes to a fast, local NVMe disk attached to the VM (\autoref{fig:system-model}). The application may be end-user facing (e.g., a database or web server) or the last hop in a networked storage system (e.g., a file server for other cloud-based applications).

\shortsection{Merkle Hash Trees}
Merkle hash trees are the standard method to protect the integrity (notably, the freshness) of storage---largely due to their proven theoretical efficiency~\cite{merkle1989certified,gassend2003caches,mckeen_innovative_2013,tsai_graphene-sgx_nodate,priebe_sgx-lkl_2020}. They play a pivotal role in ensuring boot disk integrity with Linux \textit{dm-verity}~\cite{android-dm-verity} and other emerging cloud runtimes. For storage, they are implemented as a custom block device driver that wraps a lower-level driver. The driver intercepts block I/O requests and implements the hash tree logic.

As shown in~\autoref{fig:merkle-hash-tree}, a Merkle hash tree is a balanced binary tree, with each node in the tree containing a hash value. A leaf node contains the MAC of a data block (and a cipher IV when encrypting data), and an internal node contains the SHA256 hash of the concatenation of the hashes of its two children. Internal node hashes are iteratively computed from leaf to root along the authentication path. The root hash \textit{authenticates} (i.e., asserts the correctness of) the current disk contents and is stored in a secure and non-volatile memory region that is inaccessible to an attacker~\cite{perez2006vtpm,taassori2018vault}.

There are two basic operations on a hash tree: \textit{verification} and \textit{update}. When the driver receives a block read call, the (encrypted) block data, MAC, and cipher IV are first fetched from disk. The MAC is checked for consistency against the fetched data. Next, it is verified against the root hash by recursively fetching, concatenating, and computing hashes up to the root. If the computed hash matches the known root, verification succeeds. When a block is written, a new MAC is computed (\texttt{H0}) and the hash tree is updated (\texttt{H0} ancestors) similarly. The new root hash is then saved to the secure non-volatile location. To prevent rollback attacks, the root can be \textit{sealed}---i.e., bound to a tamper-resistant counter, signed, and saved to non-volatile storage~\cite{parno2011memoir,strackx2016ariadne,angel2023nimble}.

\begin{figure}[t]
    \centering
    \includegraphics[width=\linewidth]{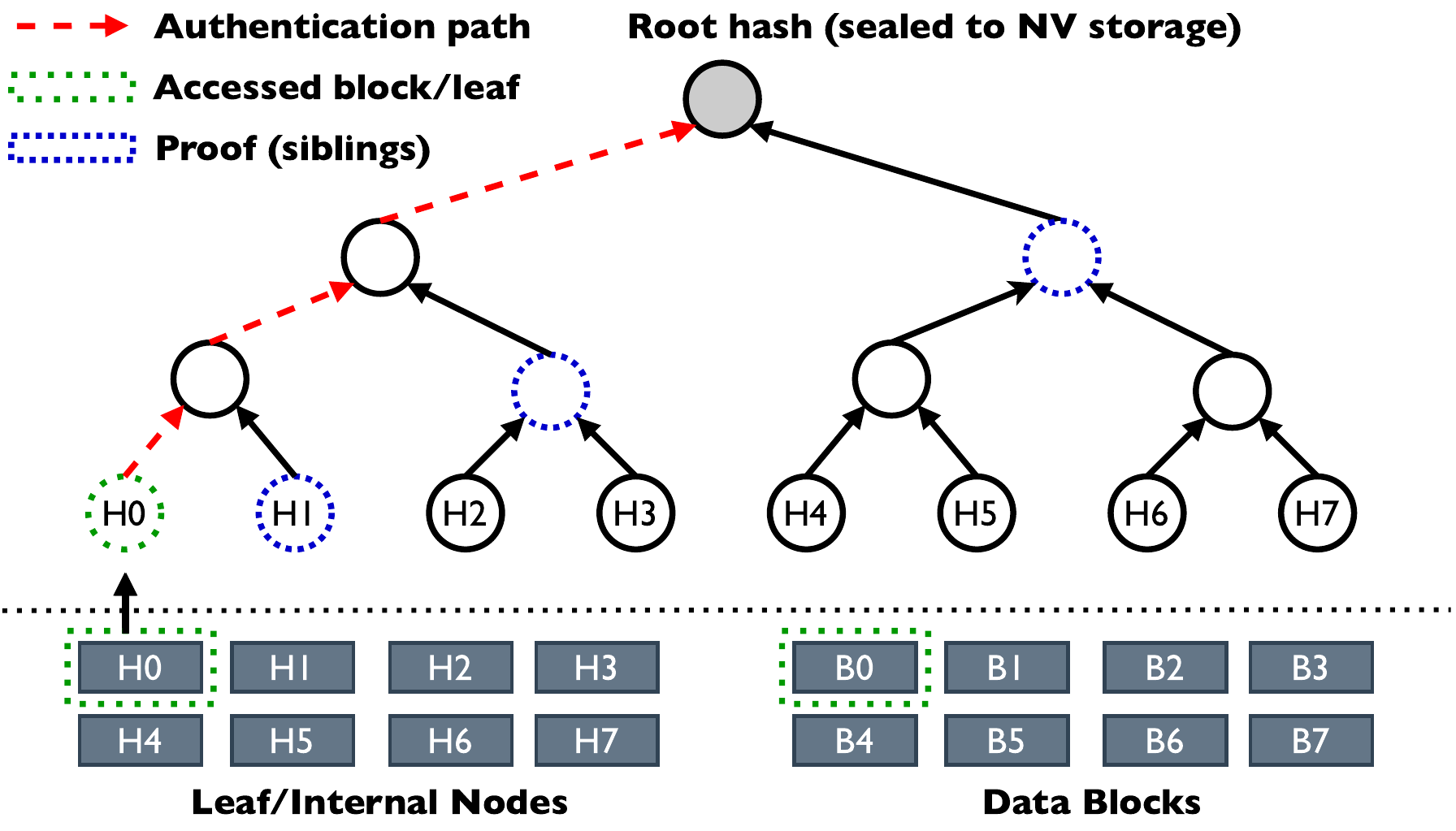}
    \caption{A Merkle hash tree protects the integrity of data read from/written to a storage device.}
    \label{fig:merkle-hash-tree}
\end{figure}

\section{Security Model}
\label{sec:sec_model}

\shortsection{Trust Model} 
We consider a standard IaaS deployment model (see~\autoref{fig:system-model}) where applications run inside guest VMs and read and write data to fast local NVMe disks. In this setting, we assume that all VM contents (code and data in memory) and the Merkle root are trusted and protected by hardware-based memory access controls. This can be ensured with confidential virtual machine technology such as AMD SEV-SNP~\cite{aws_sev_snp,gcp_confidential_vm,azure_confidential_vm}. We assume the storage devices and the hypervisor that manages access to storage (virtual or physical) are untrusted. The trusted and untrusted components have a simple block read/write interface (see \autoref{fig:sec-model}).

\shortsection{Threat Model}
We consider a privileged attacker who has access to the hypervisor or storage backbone in a public cloud data center~\cite{tsai_graphene-sgx_nodate,arnautov_scone_nodate}. This could be a malicious co-tenant with escalated privilege or a malicious cloud administrator. We assume the attacker attempts to execute \textit{data-only attacks}, attacks that are based on maliciously crafted data being returned to the block device driver by a malicious storage device; they present a significant threat to modern cloud applications~\cite{johannesmeyer2024practical,kurmus2017random,galanou2023trustworthy,bohling2020subverting,hu2016data}. The attacker has the ability to access, corrupt, swap, drop, record, inject, or replay any data across the storage interface. \newtxt{We also assume that the attacker has control over vCPU scheduling and can thus arbitrarily suspend or delay vCPU execution~\cite{li2019exploiting,schluter2024wesee}}. Note that we consider denial-of-service attacks out-of-scope, but delaying program execution can also affect integrity guarantees; we defer further discussion to~\autoref{sec:pac}.

\newtxt{We focus particularly on replay attacks; the other attacks are protected via keyed hashes (block MACs). Replay attacks} could manifest in the OS loader reading old, vulnerable versions of binaries, or file systems making incorrect access control decisions, among other attacks. Replay attacks cause externally visible effects (i.e., affect client decision-making) which are infeasible to undo and thus must be prevented.

\shortsection{Security Requirements} 
To prevent unauthorized access, the confidentiality of data can be ensured with standard AES-GCM encryption. Ensuring integrity requires guaranteeing authenticity (i.e., that some data originates from a trusted party) and freshness (i.e., that read data is the most recent version written by a trusted party). The MACs output by AES-GCM provide authenticity. Storing MACs in a secure location (after each update) then ensures freshness. However, keeping all MACs in secure memory is infeasible. A Merkle hash tree thus provides an efficient data structure to facilitate verifications and updates; it only requires keeping the Merkle root in secure memory~\cite{avanzi2022cryptographic,gassend2003caches,arasu2021fastver,matetic2017rote}.

\begin{figure}[t]
    \centering
    \includegraphics[width=\linewidth]{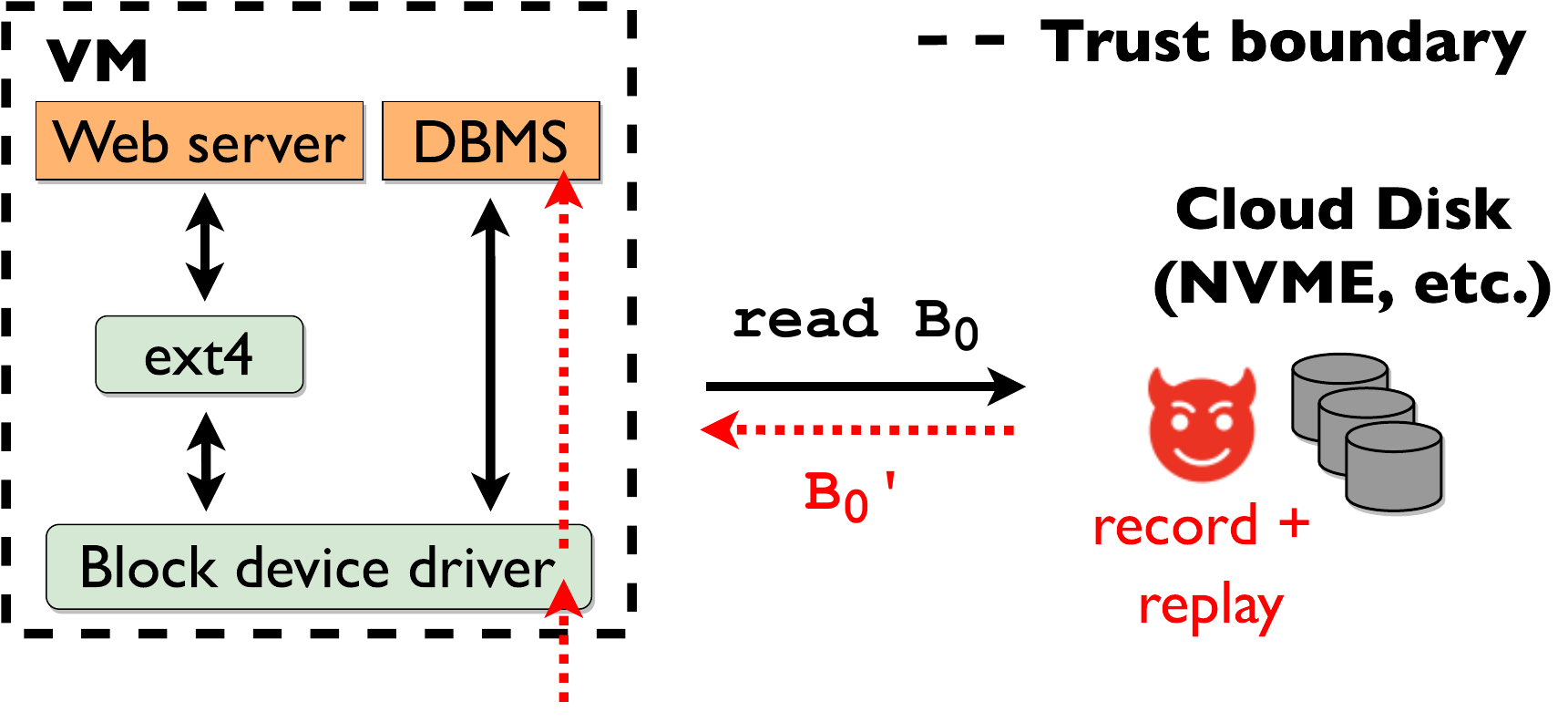}
    \caption{We assume that VM memory contents are trusted and cloud disks are untrusted; VM memory can be protected with trusted execution primitives~\cite{aws_sev_snp}.}
    \label{fig:sec-model}
\end{figure}

\section{Problem Statement}
\label{sec:motivation}

Merkle hash trees have been widely celebrated because of their proven theoretical efficiency~\cite{crosby2011authenticated}. Asymptotically, the $O(\log n)$ traversal costs have translated into low performance overheads compared to other data structures that ensure integrity. However, recent works have challenged this assertion. For example, it was shown that tree traversal costs can manifest in more than $10\times$ and $3\times$ slowdowns in the context of secure memories~\cite{taassori2018vault,feng2021scalable} and cloud storage~\cite{arasu2021fastver,arasu2017concerto,burke2025scalable}, respectively. 

Reducing performance overheads requires fundamentally rethinking the hash tree design at both the data structure- and algorithm-level. However, there is a tradeoff between security and performance. This section shows that prior works have failed to deliver a solution that provides both and motivates our search for a better approach.


\begin{figure}[!t]
    \centering
    \includegraphics[width=\linewidth]{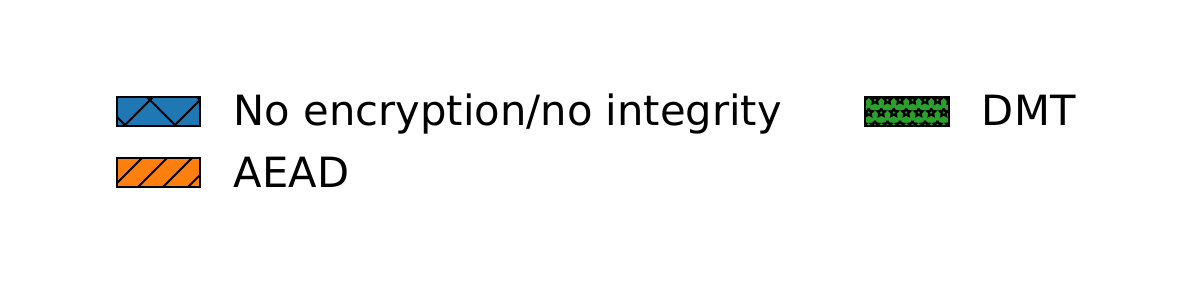}\\
    \begin{subfigure}{\linewidth}
        \centering
        \includegraphics[width=\linewidth]{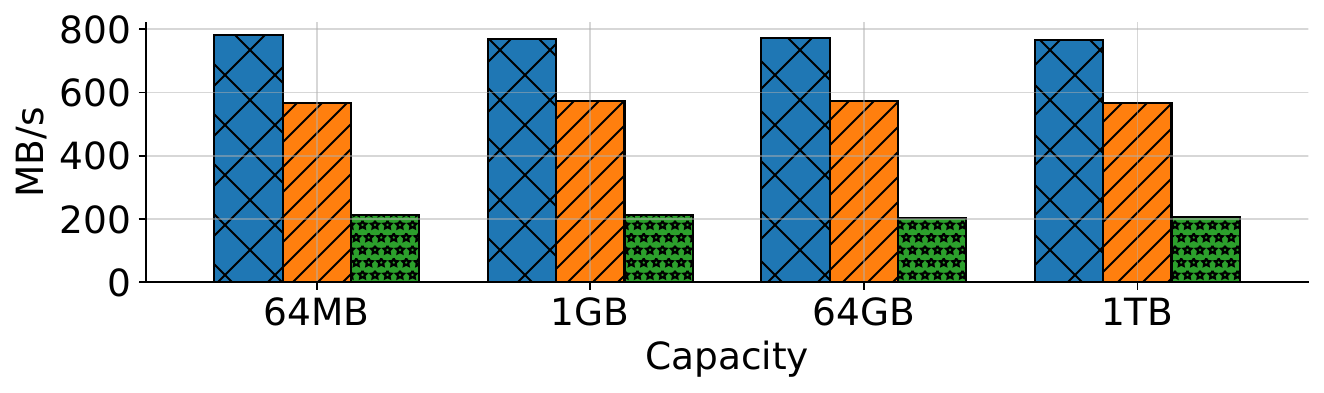}
        \caption{Aggregate read/write throughput vs. capacity.}
        \label{fig:agg_bw_cap_motivation}
    \end{subfigure}
    \hfill
    \begin{subfigure}{\linewidth}
        \centering
        \includegraphics[width=\linewidth]{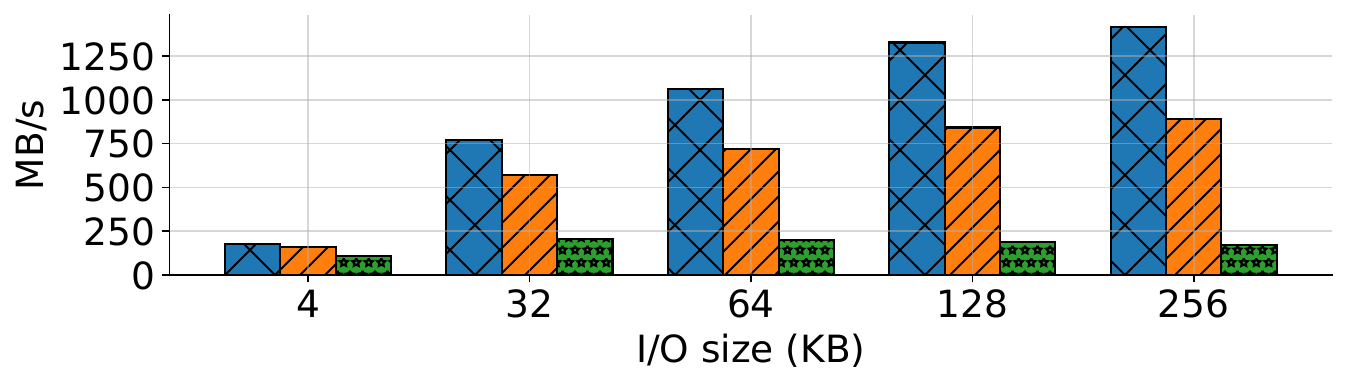}
        \caption{Aggregate read/write throughput vs. I/O size.}
        \label{fig:agg_bw_iosize_motivation}
    \end{subfigure}
    \caption{Aggregate read/write throughput vs. capacity and I/O size. Experiment parameters: Workload: Zipf(2.5), Read ratio: 1\%, Cache size: 10\%, Capacity: 1~TB, I/O size: 32~KB, Threads: 1, I/O depth: 32.}
    \label{fig:bottleneck11}
\end{figure}

\shortsection{Data Structure Optimizations}
We begin with a motivating experiment in~\autoref{fig:bottleneck11} examining the state-of-the-art hash tree design Dynamic Merkle Trees (DMTs)~\cite{burke2025scalable}. DMTs were recently proposed as alternative to the balanced, binary hash trees traditionally used to secure disks (e.g., through Linux \texttt{dm-verity}). They are based on splay trees, which are dynamic, unbalanced tree structures that self-adjust on the fly based on observed workload patterns. DMTs were shown to closely approximate an optimal tree structure. We defer implementation details to~\autoref{sec:eval}, but note that like prior works, the hash tree is implemented in a block device driver that intercepts block I/Os and performs verifications and updates when synchronously reading and writing data. 

\autoref{fig:bottleneck11} shows how the performance of DMTs changes with respect to disk capacity and I/O size. \autoref{fig:agg_bw_cap_motivation} shows that DMTs deliver 200~MB/s throughput across capacities representative of very small and very large disks. However, at 200~MB/s DMTs still observe a 67\% throughput loss relative to the AEAD (Encryption/no integrity) baseline. Note that this baseline represents an encrypted disk with no integrity protection, while DMT represents an encrypted disk that is protected by a Merkle tree. Further, disk throughput typically grows with larger I/O sizes, but \autoref{fig:agg_bw_iosize_motivation} shows that DMT performance does not scale with I/O size: performance plateaus and losses are nearly 80\% with 256~KB I/Os.

Caching hashes in secure memory (i.e., in a protected memory region) is also a standard optimization~\cite{gassend2003caches,arasu2021fastver}. Caching reduces I/O costs associated with fetching hashes during a verification or update and prompts early exits during verification (as the hashes were previously authenticated before being placed in the cache). We found that with a small cache, hit rates are consistently high (99.9\%). The implications of this are twofold. Verification costs are thus largely hidden---particularly under read-heavy workloads. However, these graphs examine write-heavy workloads, which reflect observations of real-world block storage systems~\cite{li2023depth}. For write-heavy workloads, hash caching clearly only helps to an extent, and new optimizations are needed.

\shortsection{Algorithm Optimizations}
Given the above, the most immediate question is how we can leverage algorithm-level optimizations to improve performance. A natural starting point is to use \textit{asynchronous execution}, which is a  building block of modern storage systems~\cite{didona2022understanding,lee_asynchronous_nodate}. Intuitively, asynchronous execution means that hash tree operations are decoupled from data operations and have minimal performance impact on the critical path. This can be achieved through means such as batched processing.

In a batching scheme, whenever data is read or written, a request object for a verification or update is created and placed on a queue. Whenever the queue hits the specified $batch\_size$ limit, the application thread blocks (i.e., \textit{checkpoints}) until all pending verifications and updates are complete. This can improve performance by removing hash tree operations from the critical path entirely. Costs can then be amortized by accumulating and processing large batches of requests at a more opportune time. An overview of this process is shown in~\autoref{fig:batching}.

\begin{figure}[!t]
    \centering
        \includegraphics[width=0.9\columnwidth]{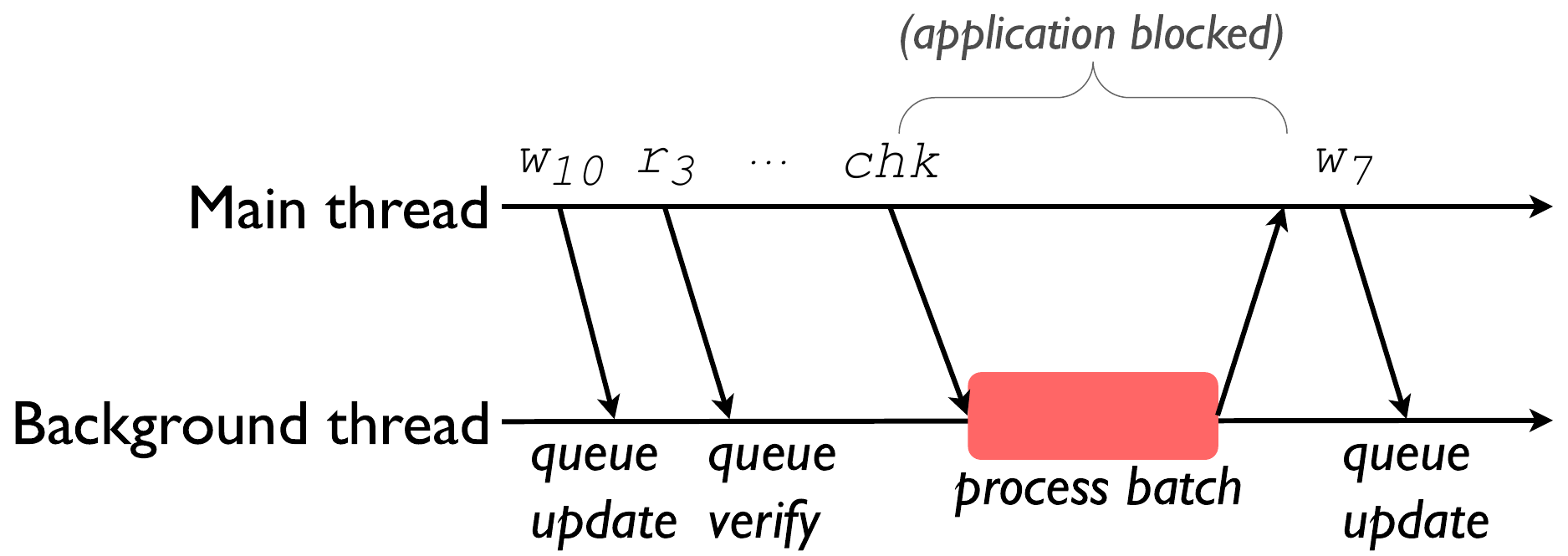}
        \caption{Batched processing (Note: \texttt{chk}=checkpoint step).}
        \label{fig:batching}
\end{figure}

Recent works have begun to embrace this idea to optimize integrity checks in real-world key-value stores~\cite{arasu2021fastver,arasu2017concerto}. However, batching introduces a significant security risk. Notably, reads always return unverified data up the call stack, and thus data-only attacks are not immediately detectable. This directly violates integrity (freshness) guarantees. 

To quantify the extent to which these vulnerabilities can be practically exploited, \autoref{fig:batch} characterizes the window of vulnerability (WoV) introduced by batching. The window of vulnerability describes the time between when data was returned to the caller and when it was actually verified. There are two key observations. First, \autoref{fig:tp-vs-batch} shows that throughput improvements are only observed with large batch sizes ($>10^5$). Moreover, \autoref{fig:wov-vs-batch} shows that the mean window of vulnerability grows with a power law w.r.t. batch size. At a batch size of $10^5$ blocks, the mean window of vulnerability is 800~$ms$, while at a batch size of $10^6$ it is 8 seconds. For reference, FastVer~\cite{arasu2021fastver} uses a batch size of 4~M requests. There is a clear security-performance tradeoff here: batching only helps with large batch sizes, but large batch sizes introduce long windows of vulnerability. 

\begin{figure}[!t]
    \centering
    \begin{subfigure}{\linewidth}
        \centering
        \includegraphics[width=\linewidth]{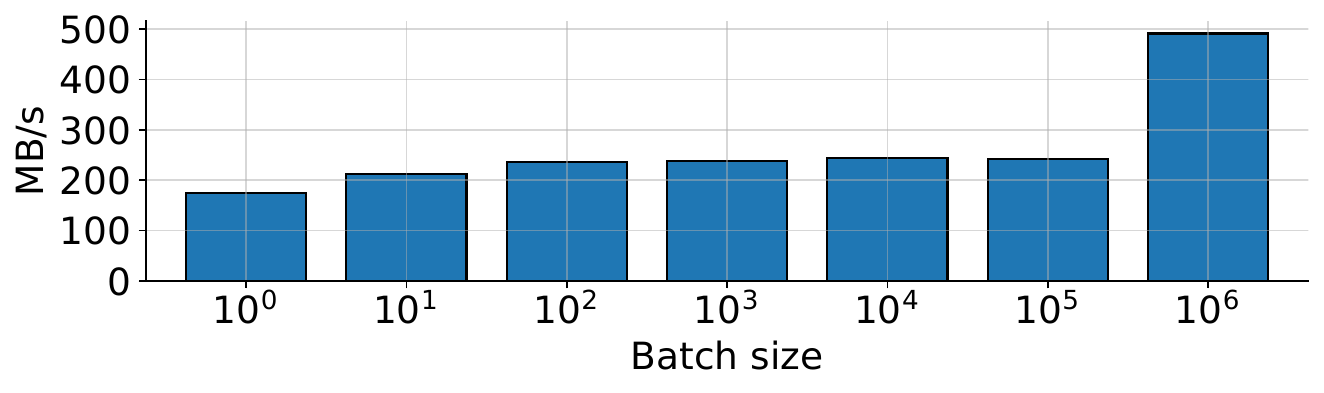}
        \caption{Aggregate read/write throughput vs.batch size.}
        \label{fig:tp-vs-batch}
    \end{subfigure}
    \hfill
    \begin{subfigure}{\linewidth}
        \centering
        \includegraphics[width=\linewidth]{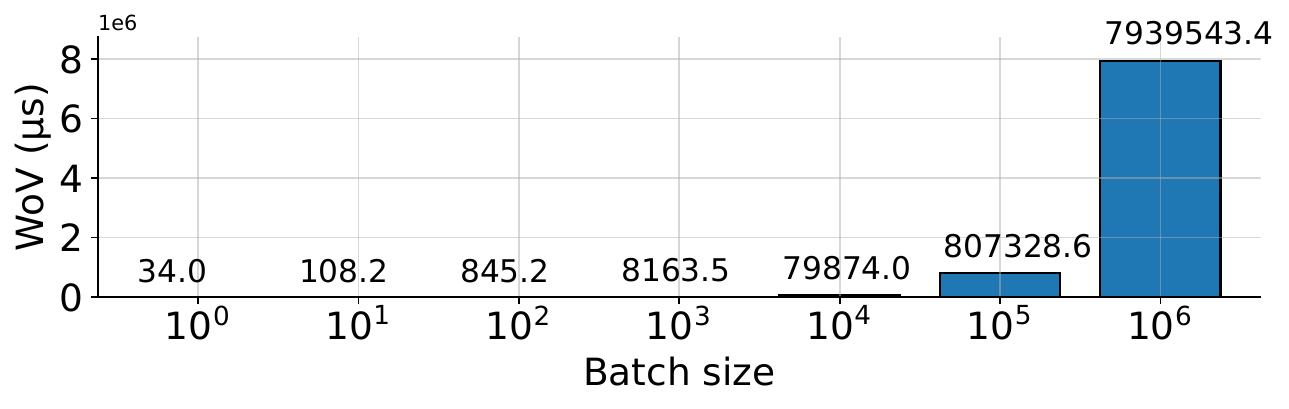}
        \caption{Window of vulnerability (WoV) length vs. batch size.}
        \label{fig:wov-vs-batch}
    \end{subfigure}
    \caption{Meaningful performance gains are only realized with larger batch sizes, but large batch sizes introduce long windows of vulnerability.}
    \label{fig:batch}
\end{figure}

\shortsection{Our Objective}
These prior works have led the community to an understanding that in order to achieve storage integrity one must either give up performance or weaken security guarantees. We posit that this tradeoff can be avoided.
This leads us to the question: \textit{Can we capture the performance advantages of asynchrony without sacrificing integrity guarantees?} In the following, we show that it is indeed possible.

\section{Partially Asynchronous Integrity Checking (\tool)}
\label{sec:pac}

\begin{table*}[t]
    \centering
    \normalsize
        \caption{Comparison of security guarantees and core mechanisms of different approaches to integrity checking. Broadly, whether each approach provides runtime consistency (integrity) checks and/or rollback protection depends on when updates and verifications are executed. For example, Batching (i.e., DMT + Batching) does not provide runtime consistency because verifications/updates are asynchronous (and thus replay detection is deferred), and achieves good performance. \tool{} (i.e., DMT + \tool{}) provides runtime consistency via synchronous verifications and rollback protection, while leveraging asynchronous updates to improve performance.}
    \begin{threeparttable}
    \begin{tabularx}{\linewidth}{Xccccc}
        \toprule
        & AEAD& \texttt{dm-verity}~\cite{android-dm-verity}\tnote{1}&DMT~\cite{burke2025scalable} & DMT + Batching & DMT + \tool  \\
        \midrule
        Runtime consistency & \emptycirc{}&\fullcirc{}& \fullcirc{}& \halfcirc{}\makebox[0pt][l]{\tnote{2}} & \fullcirc{}  \\
        Rollback protection & \emptycirc{}&\emptycirc{}& \fullcirc{} & \fullcirc{} & \fullcirc{}  \\
        Updates (writes) & \textit{Sync} & \textit{Sync}& \textit{Sync} & \textit{Async} & \textit{Async}  \\
        Verifications (reads) & \textit{Sync}& \textit{Sync} & \textit{Sync} & \textit{Async} & \textit{Sync}  \\
        \bottomrule
    \end{tabularx}
    \begin{tablenotes}
        \item[1] This denotes \texttt{dm-verity} variants that support writes.
        \item[2] Batching has deferred runtime consistency checking.
  \end{tablenotes}
    \end{threeparttable}
    \label{tab:properties}
\end{table*}

This section shows how we can use concurrency to maintain the security guarantees of a synchronous hash tree design while exploiting the performance advantages of an asynchronous hash tree design.

\shortsection{Security Guarantees}
To provide integrity on untrusted storage, we design our system around two security guarantees:

\vspace{1em}

\begin{enumerate}
\itembase{6pt}
    \item \textbf{Read Guarantee:} The integrity of read data is always verified before the data is returned to applications; reads are never speculative. This ensures that applications will always receive correct and fresh data and closes the window of vulnerability entirely. 
    \item \textbf{Write Guarantee:} Hash tree updates are executed asynchronously, but the main thread coordinates with the background thread to ensure liveness of hash tree updates (and state updates). This prevents rollback attacks introduced by asynchrony. The sealed Merkle root always accurately reflects the latest \texttt{FSYNC} call.
\end{enumerate}

\subsection{\tool~Overview}
We introduce a \underline{p}artially \underline{a}synchronous integrity \underline{c}hecking (\tool) method. \tool~exploits asynchronous execution through concurrency: it uses a background thread that assists in the verification and update process. The design centers around a shared queue that resides in secure memory and holds pending update requests. It is driven by two key mechanisms: \textit{synchronous verifications} and \textit{flush-consistent asynchronous updates}. Note that these mechanisms are transparent to applications and handled entirely within the block device driver. An overview is shown in~\autoref{fig:pac} and comparison to prior approaches is shown in~\autoref{tab:properties}.

Synchronous verification means that data is always verified by the main thread before being returned to the caller; these costs are minimized via the hash cache. In contrast, update requests are queued by the main thread and executed asynchronously by the background thread. Updates are executed one at a time. The background thread polls the queue at a configurable rate; we use a thread sleep to enforce this. Flush-consistency means that when the driver receives a device \texttt{flush} call (triggered via an \texttt{FSYNC} system call), the main thread blocks until all pending updates are complete. As device flushes are used to provide durability guarantees to applications (i.e., that data is in fact persistent and not residing in the OS page cache or disk hardware caches), this ensures that the Merkle root always accurately reflects what disk contents the application believes to be durable.

The result is that: (1) reads are never speculative, and (2) the main thread can largely proceed with other useful work (e.g., other file system or application logic) while the background thread processes update requests. Realizing this requires several additional support structures.

\begin{figure}[!t]
    \centering
    \includegraphics[width=0.5\textwidth]{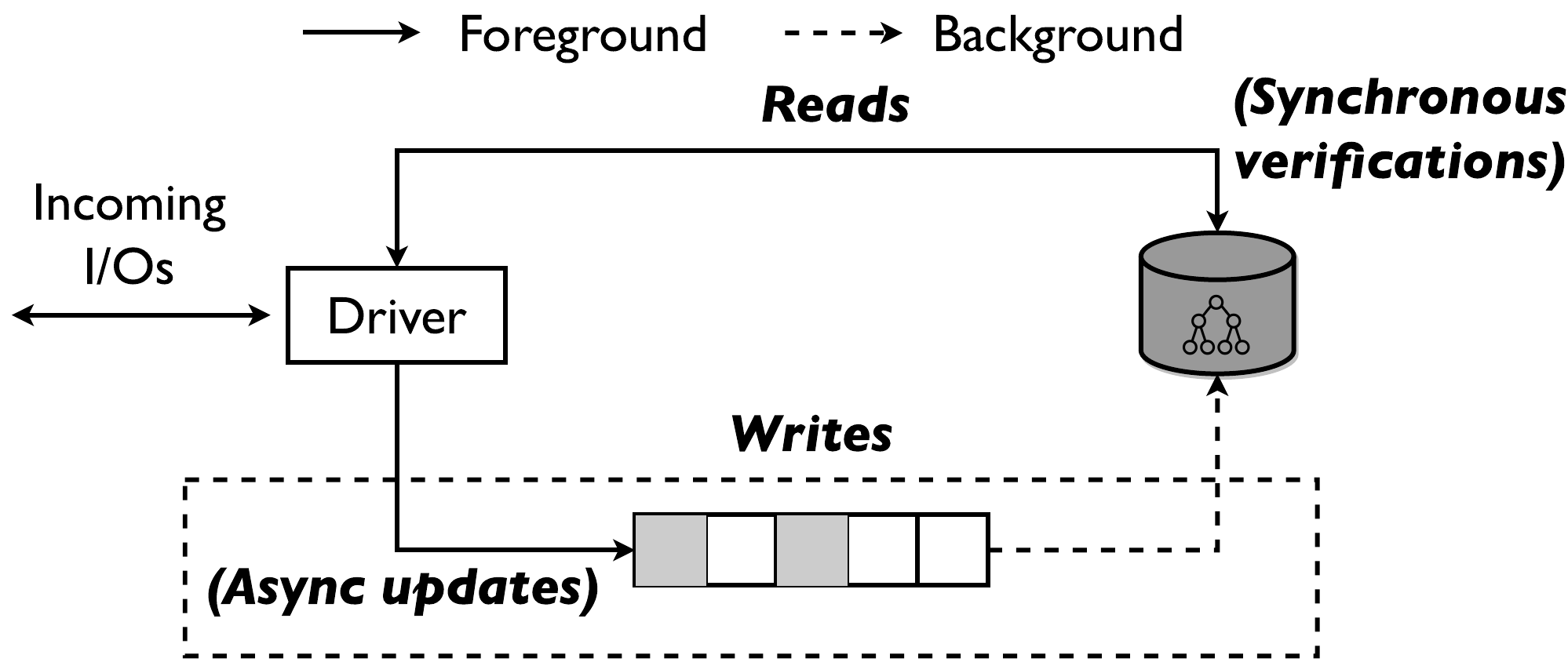}\\
    \caption{\tool~overview. Verifications are executed synchronously by the main thread, while updates are queued by the main thread and executed asynchronously by the background thread. We refer to~\autoref{fig:rwpath} for a timeline diagram of further technical details.}
    \label{fig:pac}
\end{figure}

\subsection{Write Path}

\shortsection{Block Writes}
When the driver receives a \texttt{write} call, the block data is encrypted and a new MAC is generated. This happens in the context of the main thread. A request object is then created and appended to the shared queue. The object contains the Merkle tree node object, new MAC, and new cipher IV. The actual block data is immediately written out to disk (but not flushed through the disk cache, i.e., made durable). This process is shown in~\autoref{fig:rwpath}.

For example, when the driver receives a write for block 10 ($w_{10}$), the main thread writes the data to disk, queues an update request, and returns execution to user space while the background thread processes the update. Later, the background thread will raise an alert if the update failed.

During the write call, \tool~must also handle the scenarios where the incoming write is to a block that still has a pending update, or where the queue is full.

\shortsection{Overriding Updates}
If during the write call there is a pending update request for the same block, the main thread overrides the pending request with the incoming one. Thus, only the latest update to a block is tracked in the queue and executed by the background thread. An alternative would be to queue each update request individually, which would require additional memory for the queue and additional compute resources to complete updates. \tool~avoids the additional memory requirement and doing wasted work on earlier update requests which were effectively voided by more recent updates. 

\shortsection{Rate Limiting Updates}
In a real deployment, the size of the queue will be capped to some percentage of the system memory capacity to more efficiently allocate resources to running applications. To handle a finite queue capacity, \tool~implements a rate limiting mechanism. During a write call, if the queue is already full, the main thread blocks until at least one free queue slot opens up. At the same time, the background thread enters a draining state where it ignores the sleep timer and executes updates as quickly as possible. Thus, when the queue becomes full, performance falls back to (in essence) doing updates synchronously, because the main thread will block while an update completes and opens a queue slot.
The queue size and high/low thresholds that control rate limiting are configurable by a system administrator and can be changed dynamically at runtime. 

\subsection{Read Path}

\shortsection{Block Reads}
When the driver receives a \texttt{read} call, the main thread first reads the data from disk. It then checks the queue for the latest version of the block MAC. If there is a pending update request, it verifies that the MAC (in the queued request) is consistent with the fetched data then decrypts the data. Otherwise, it proceeds with the normal read/verify procedure: it fetches the node from the hash cache (or disk), decrypts the data with it, then verifies the MAC in the hash tree. This process is shown in~\autoref{fig:rwpath}.

For example, when the driver receives a read call for block 3 ($r_{3}$), there is no pending update request for block 3 in the queue, so the main thread proceeds with the normal read/verify procedure. However, when the driver later receives a read call for block 7 ($r_{7}$), there is still a pending update request for block 7, so the main thread pulls the MAC directly from the queued request. If the data read from disk is consistent with this MAC, verification is complete and the data is returned to the caller---there is no need to traverse the hash tree. Thus, the queue can also prompt early exits during verification like the hash cache can.

\begin{figure}[!t]
    \centering
    \includegraphics[width=0.45\textwidth]{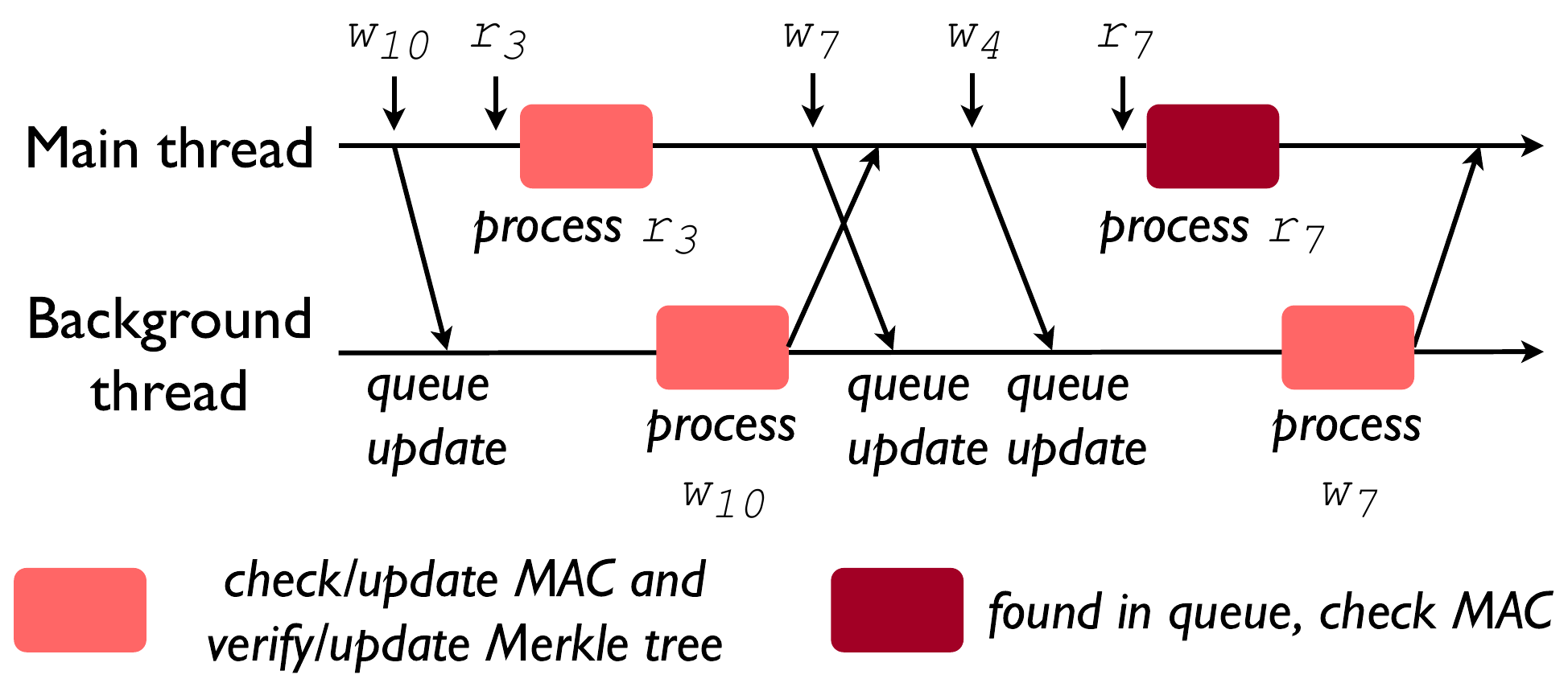}\\
    \caption{In \tool, verifications are executed synchronously to ensure freshness on reads, while updates are executed asynchronously in a background thread to exploit parallelism.}
    \label{fig:rwpath}
\end{figure}

\subsection{Rollback Protection}
While our focus is not rollback attacks, executing updates asynchronously complicates rollback guarantees. This section describes how \tool~mitigates rollback vulnerabilities introduced by using concurrency.

\shortsection{Rollback Challenges}
Standard methods to provide rollback protection and recovery are known. For example, sealing the Merkle root (the disk \textit{state}) with a tamper-resistant counter can prevent an attacker from using a crash to ``reset'' the Merkle root and replay an old version of the entire disk to applications~\cite{parno2011memoir,strackx2016ariadne,angel2023nimble}. However, privileged attackers still \newtxt{have control over vCPU scheduling and other operating system tasks}~\cite{li2019exploiting,schluter2024wesee}. If hash tree updates (and state updates) are handled by a separate thread, the attacker can arbitrarily delay the update thread to control the degree to which state updates are sealed. In the worst case, they can prevent the sealed state from ever advancing past the initial state, and after a crash the system would simply recover to the initial state. The problem is that there must be a \textit{checkpointing} mechanism through which the main thread can validate that hash tree updates are in fact progressing in the background (i.e., liveness guarantee).

\shortsection{Coordinated State Updates}
To design the checkpointing mechanism, we must first consider how to securely manage \textit{state}. For storage, state consists of the set of disk blocks, but is represented by the Merkle root. Rollback protection can be ensured by sealing the state with a tamper-resistant counter. 
However, storage systems are presented an additional challenge with regards to managing state: while in theory the state is a small 32~B Merkle root, updating state is effectively a two-step process: all hash tree updates must first be reflected in the Merkle root, and the final Merkle root must then be sealed. We refer to the former as the \textit{preparation} phase and the latter as the \textit{sealing} phase.

\begin{figure}[!t]
    \centering
    \includegraphics[width=0.45\textwidth]{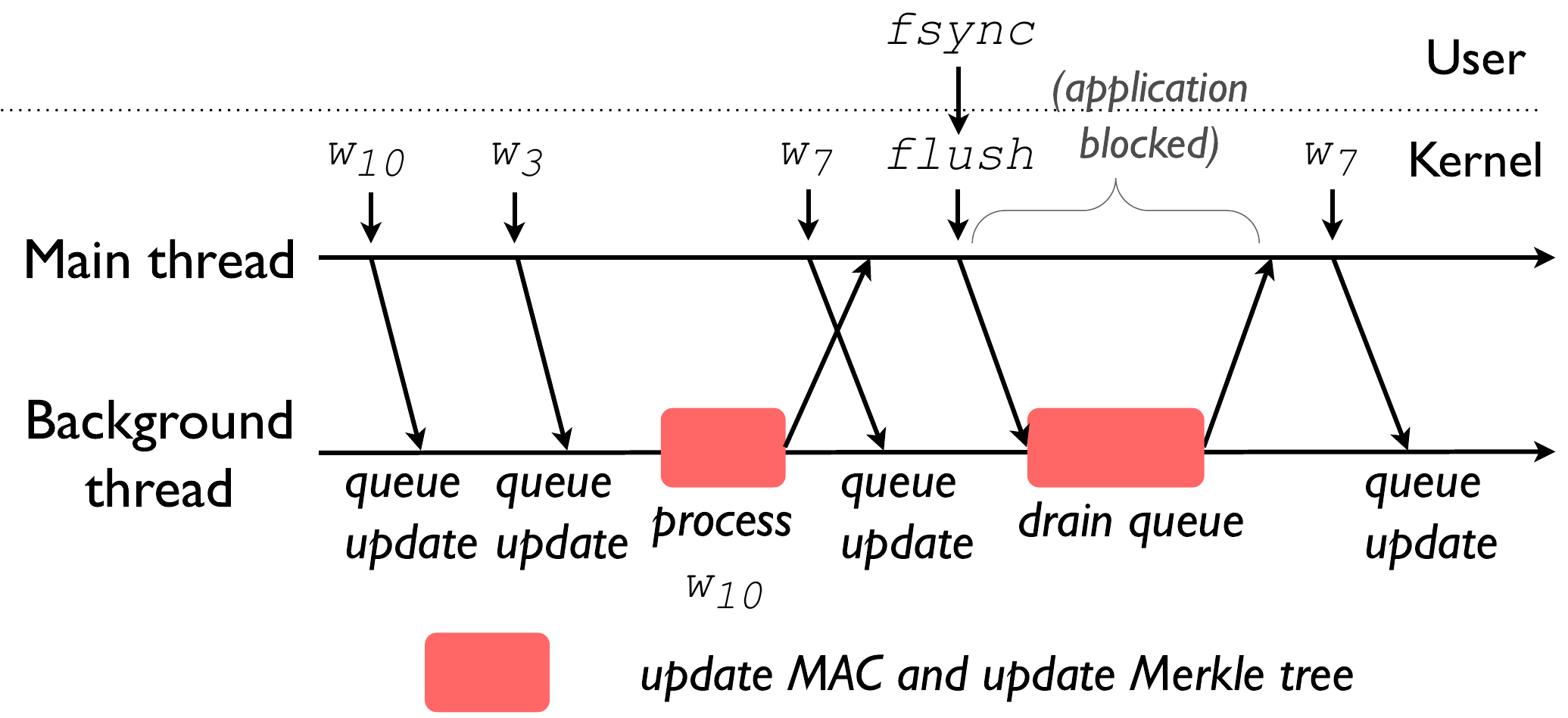}\\
    \caption{In \tool, device flush commands (triggered by an \texttt{FSYNC} system call) force the main thread to block while all pending update requests complete and a state update is committed (i.e., Merkle root is sealed). This ensures liveness of the background thread and rollback protection.}
    \label{fig:fsyncpath}
\end{figure}

In \tool, both threads coordinate in performing state updates. The background thread handles preparation. Sealing is invoked at the discretion of the application and can be handled in either the main thread or background thread. However, the main thread checkpoints by blocking in the device flush function until both preparation and sealing are complete. This is shown in~\autoref{fig:fsyncpath}. \newtxt{Note that without blocking, the main thread would have no guarantee that the Merkle root was updated and sealed properly. Thus, the checkpoint mechanism guarantees both liveness and rollback protection: any subsequent integrity check against the Merkle root must either be fresh or fail.}

Checkpointing during a flush (triggered by an \texttt{FSYNC} system call) is desirable because it is a natural synchronization point upon which crash-consistency protocols are built; it signifies a clear intent to commit a state update. From an application point of view, data (whether buffered in memory or written out to disk) is not considered persistent until an fsync is called and a flush command is sent to the device. Upon a crash, updates queued since the last flush would \newtxt{indeed be lost, but alongside the dirty (buffered) data as well. We note that ideally applications should flush data in a timely fashion. However, reducing (expensive) flushes is common practice when tuning application configurations. As an additional measure of protection, flushes could be opportunistically (or eagerly) induced in the driver.}

\section{Security Analysis}
\label{sec:sec_analysis}

In this section, we show that \tool~provides strong integrity guarantees (e.g., prevent replay and rollback attacks) during normal operation and following system crashes.

\begin{theorem}
     In the absence of system crashes, \tool~always returns the most recent version of data to the caller on reads or the integrity check fails. ({\it Read guarantee})
\end{theorem}

\begin{proof}
Let $R_i$ denote the Merkle root at version $i$. Without loss of generality, consider the block at address $j$. We denote with $B$ and $h(B)$ the current data stored in the block at address $j$ and its hash, respectively. Now, assume that a write operation updated the block at address $j$ to the value $B'$.  We define the most \textit{recent} data ($B'$ here) to be the version that reflects the latest write issued to the driver (but not necessarily flushed through disk caches, i.e., durable).

When the block was updated, the data is sent to be written out to disk and an update request was placed on the queue. Upon the read request, one of two scenarios will occur. Either the update will have completed, in which case there will not be a pending update in the queue and the new block hash $h(B')$ would be reflected in a new Merkle root $R_{i+1}$. In this case, the driver will read the updated block $B'$ from disk, fetch $h(B')$ from disk (or the hash cache), proceed with the normal recursive verification process, and return $B'$ if it succeeds or an integrity check fault if not. Otherwise, if the update had not completed, then there will be a pending update request in the queue and the driver will read the updated block $B'$ from disk, check the consistency of the block data with the queued hash $h(B')$, then return $B'$ if it succeeds or an integrity check fault if it not.

The most recent version of the data $B'$ is always reflected in the Merkle tree or in the queue.
Thus, for an attacker to successfully replay data they would have to present a value $B^*$ and hash $h(B^*)$ that either passes the verification check against the new Merkle root $R_{i+1}$ or the queued hash $h(B')$. Because of the collision resistance property of the cryptographic hash function used, the adversary is unable to find $B^*$ such that, $h(B^*) = h(B')$. Further, hardware-based memory access controls (recall our trust model in \autoref{sec:sec_model}) ensure that the Merkle root is sealed and that it can not be modified by the attacker to facilitate the attack. Thus, the adversary could not provide a block $B^*$ that would be verified against stored hash values without triggering an integrity check fault.
\end{proof}

\begin{theorem}
    After a system crash, \tool~always recovers the most recent state perceived as durable by the application or the integrity check fails.  ({\it Write guarantee})
\end{theorem}

\begin{proof}
Let $S_i$ denote the disk state (set of all block data) at version $i$ and $R_i$ denote the corresponding Merkle root that is sealed during a \texttt{FSYNC} call to commit (make durable) this disk state $S_i$. Assume a crash occurred, where the most recent sealed root $R_{k}$ is returned to the system during recovery. An attacker would not be able to produce a set of data $S_{k*} \ne S_{k}$ that would hash to the protected and sealed root $R_{k}$ because of collision resistance of the hash function. Thus, the recovery process will always return data of the most recent state made durable or the integrity verification process will fail.
\end{proof}

\section{Performance Evaluation}
\label{sec:eval}

In this section we evaluate the performance of~\tool{} across several realistic workloads. We compare \tool{} with state-of-the-art hash trees DMTs~\cite{burke2025scalable}~\footnote{\newtxt{We note that \tool{} is a general approach and the underlying data structure can be swapped out for alternatives (e.g., list-based ones).}} and report results against two non-integrity verified baselines: No encryption/no integrity and AEAD (Encryption/no integrity). Our evaluation focuses on the following research questions:

\vspace{6pt}

\begin{rqlist}
\itembase{3pt}
    \item \label{rq:fsync} What overheads does \tool{} incur and under what conditions?
    \item \label{rq:size} Is \tool~scalable (i.e., broadly maintains a stable performance guarantee)?
    \item \label{rq:storage} What memory and storage costs are associated with \tool?
\end{rqlist}

\subsection{Experimental Setup}

\shortsection{Implementation}
We implement the hash trees in 5k lines of C++. We use BDUS to implement a custom block device driver that wraps a lower-level driver~\cite{faria_bdus_2021}. BDUS has a kernel module that exposes block layer hooks to user space~\footnote{\newtxt{We note that other frameworks could be used to implement \tool{}, including pure kernel, pure user space, or hybrid approaches. There are inherent trade-offs to each (e.g., less isolation and larger attack surface in kernel space); this is an orthogonal issue that is an active area of research~\cite{jacobs2024system,yasukata2023zpoline}.}}. The three functions of interest are \texttt{read()}, \texttt{write()}, and \texttt{flush()} which are invoked by the kernel whenever a block is read, a block is written, or an \FSYNC{} is called by the application. Our basic data unit aligns with the disk I/O size (4~KB blocks)~\cite{khati2017full,brovz2018practical,nist-aes-xts}. Note that best-known methods have shown that sealing can be done within single-digit milliseconds for small state sizes ($<1$~KB)~\cite{niu2022narrator,angel2023nimble}. We simulate a state update by putting the main thread to sleep for 5~$ms$ at the end of the device flush function. A parameter overview is shown in~\autoref{tab:parameters}.

Like prior works, we ensure authenticated encryption with AES-GCM~\cite{avanzi2022cryptographic,taassori2018vault}. We use a 128-bit encryption key for blocks. The MACs produced during encryption are used as the leaves in the tree. For internal nodes, we compute 256-bit hashes using SHA-256 with a 256-bit key.

\shortsection{Testbed}
We perform all experiments on a cloud server equipped with a 48-core 2.8~GHz AMD EPYC 7402P processor, 128~GB memory, and \newtxt{two local NVMe SSDs (one for data, one for metadata). Note that where metadata resided did not significantly impact results; small hash caches tend to be very efficient.} 
We reinitialize hash trees between each experiment and use a standard LRU cache replacement policy. For DMTs and consistent with prior evaluations, we set the splay window flag $w=True$ and splay probability $p=0.01$, which means that splaying only occurs 1\% of the time. For \tool, we set the queue size to 1024 and the low watermark threshold for rate limiting to 0.75.

\begin{table}[t]
    \centering
    \small
    \caption{Key experimental parameters.}
    \begin{tabularx}{\linewidth}{Xl}
        \toprule
        \textbf{Parameter} & \textbf{Description} \\
        \midrule
        \textit{Capacity} & Usable capacity for data blocks \\
        \textit{Cache size} & Cache size as \% of tree size \\
        \textit{Read ratio} & \% of read operations \\
        \textit{I/O size} & Size of application I/O \\
        \textit{I/O depth} & Max no. outstanding application I/Os \\
        \textit{Thread count} & Number of application threads \\
        \textit{Background rate} & Polling rate of the background thread \\
        \textit{\FSYNC{} period} & Duration between application flush calls \\
        \bottomrule
    \end{tabularx}
    \label{tab:parameters}
\end{table}

\shortsection{Workload Settings}
We adopt (and extend) workloads and evaluation strategies similar to those used in prior work on storage integrity. Here we generate workloads with \textit{fio}~\cite{fio}. Workloads have a 5-minute warm-up period and 15-minute benchmark period. Unless stated otherwise, mean latency or throughput is reported in each graph.

When selecting workloads, we must consider how \tool~would be deployed in a real system. Exactly when the read, write, and flush functions are called depends on how the application interacts with the disk. It may read and write directly to the raw device file---e.g., as the PostgreSQL database does. In contrast, it may interact with the disk through a file system, which is what many applications use to store unstructured data. In either case, the application might bypass the OS page cache (by opening the regular file or device file with the \texttt{O\_DIRECT} flag) or rely on the page cache for buffered I/O. In the latter case, the sizes and timing of block I/Os issued to the driver depend on how often the application explicitly calls \FSYNC{} (to write out dirty pages), or if memory pressure causes the kernel to begin writing dirty pages back to disk. 

We examine the performance space under these different scenarios by running experiments directly against the storage device and varying key system and workload parameters (capacity, cache size, etc.). \newtxt{This allows us to provide generalizable insights on PAC performance and offer guidance on steering applications towards optimal configurations (e.g., tuning flush frequency).} We toggle the \texttt{O\_DIRECT} flag in \textit{fio} to do this, per best practice~\cite{google_pd_benchmarking}. We focus especially on write-heavy workloads with a Zipfian shape, which closely approximate real-world storage access patterns~\cite{li2023depth,yang2016write}. 

Unless stated otherwise, default parameters include---Read ratio: 1\%, I/O size: 32~KB, Thread count: 1, I/O depth: 32, Capacity: 1~TB, Cache size: 10\%, Target background rate: 1000 requests/s, \FSYNC{} period: 1000 updates. These parameters showcase the best performing configuration for the baselines and reflect the shape and behavior of real-world storage workload patterns~\cite{li2023depth,yang2016write}.

\subsection{Results}
\subsubsection{\tool{} Performance}
\label{sec:parallelism}

\tool{}'s performance is affected both by its configuration parameters (key among these being target background processing rate) and the workload's tolerance of uncommitted writes (represented in our workloads as the \FSYNC{} period). We briefly evaluate and optimize system parameters, followed by a performance evaluation of \tool{} on diverse workloads.

\begin{figure}[t]
    \centering
    \includegraphics[width=\linewidth]{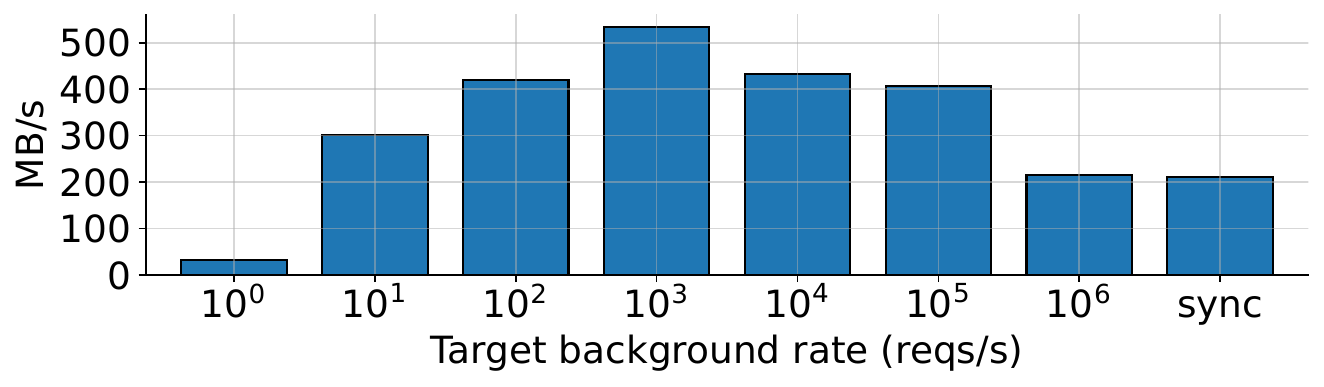}
    \caption{Aggregate read/write throughput vs. target background processing rate. ``sync'' denotes synchronous DMT, and the remaining data points denote DMT+\tool~at various rates. There is a region of optimal performance between 100 and 100k~requests/s, with the maximum throughput observed at 1k~requests/s.}
    \label{fig:tp_vs_bgrate}
\end{figure}

\shortsection{Impact of Target Background Processing Rate}
The target background processing rate describes the target rate that the background thread polls the queue and executes updates. \autoref{fig:tp_vs_bgrate} shows that maximum throughput is achieved at 1k~requests/s, with a $2.3\times$ speedup over synchronous DMT (the bar labelled ``sync''). At a rate of 1M~requests/s, the throughput is approximately equal to synchronous DMT.

The target rate that the background thread polls the queue influences the degree of rate limiting and update overriding. If it polls too slowly (i.e., the polling rate is less than the workload IOPS), the queue quickly accumulates update requests and becomes full. This leads to a significant amount of rate limiting to free up queue slots. Moreover, the background thread remains largely idle when it could have been polling faster. In contrast, if it polls too quickly, updates tend to be processed very quickly, which reduces the opportunity to override updates. If the processing of certain updates had been delayed for slightly longer, this wasted work could have been eliminated. 

This manifests in \autoref{fig:tp_vs_bgrate} as a region of optimal performance between 100 and 100k~requests/s. To understand why, consider that synchronous DMT has a throughput of approximately 200~MB/s, which amounts to approximately 6100 read/write IOPS for 32~KB I/Os. This workload is 99\% writes, so 6k of those I/Os were writes (which prompted an update request to be queued). A target rate of 1k~requests/s balances this tradeoff best: it is sufficiently fast to prevent the queue from becoming full too quickly, but still gives the background thread sufficient time to exploit reference locality by overriding updates.

\shortsection{Impact of \FSYNC{} Period}
We fix the target background processing rate at 1k requests/s and next examine how the \FSYNC{} period influences performance. The \FSYNC{} period describes the duration (in number of write I/Os) between \FSYNC{} calls. It is important because two critical events happen during device flushes---the queue drain and Merkle root sealing. Analyzing how throughput changes w.r.t. the \FSYNC{} period allows us to characterize to what degree parallelism helps when applications issue \FSYNC{} calls more or less frequently.

\begin{figure}[t]
    \centering
    \includegraphics[width=\linewidth]{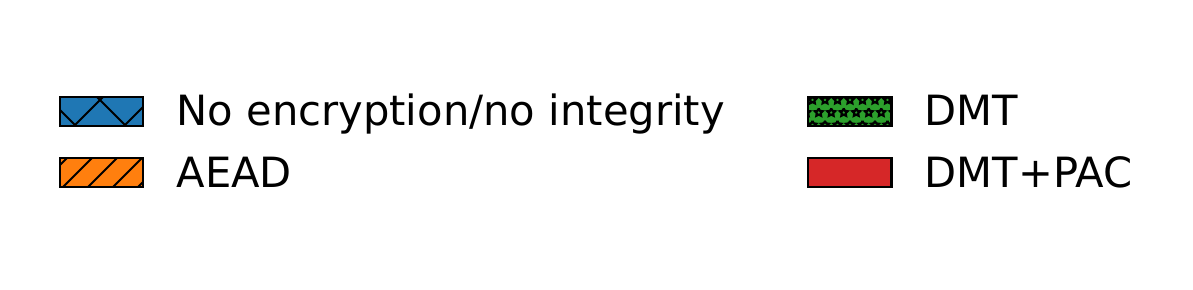}
        \includegraphics[width=\linewidth]{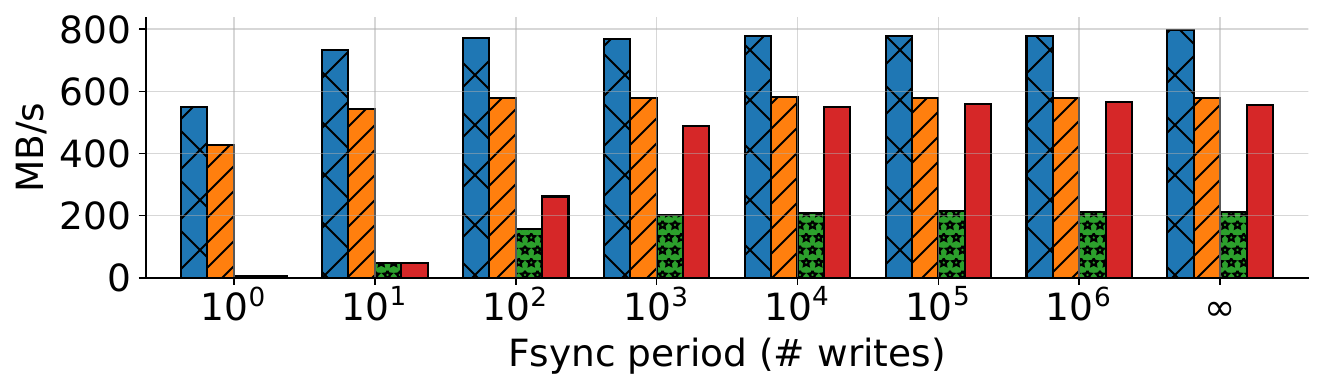}
    \caption{Aggregate read/write throughput vs. \FSYNC{} period. Longer periods between flush calls enable \tool~to better exploit parallelism.}
    \label{fig:tp_vs_fsync}
\end{figure}

\autoref{fig:tp_vs_fsync} shows performance of each integrity technique as workload \FSYNC{} period varies. \tool{} outperforms synchronous DMTs when \FSYNC{}s occur every 100 writes or more, with performance $>85\%$ of AEAD when \FSYNC{}s occur every 1000 writes or more. In effect, when \FSYNC{}s are sufficiently infrequent, \tool{} virtually eliminates the overhead of integrity checking. When the \FSYNC{} period is $\leq 1000$ updates, flushes are more frequent so the background thread executes most updates while draining the queue during the flush. Thus, there is little opportunity for the background thread to amortize update costs. 

These observations imply that the background thread effectively races the main thread to amortize update costs. If the driver is kept sufficiently busy (i.e., $\geq 1000$ update requests) between flushes and flushes are infrequent, then update (and sealing) costs can largely be amortized. 
However, if the driver is starved of I/O requests (e.g., because they are still sitting in the page cache) or flushes are frequent, \tool~falls back to (in essence) doing updates synchronously. 

\newtxt{It is worth noting that prior works have tackled this problem (in the synchronous case) with \textit{batched state updates}~\cite{niu2022narrator,angel2023nimble}, which artificially increases the \FSYNC{} period, but weakens security guarantees. We argue that this is not only insecure, but unnecessary. Rather than batching state updates, we can reframe the \FSYNC{} period as a principled tuning knob for applications---maintaining strong security guarantees and eliminating overheads by instead restructuring the way in which applications issue dirty page writebacks and flushes~\cite{larabel2024rwf_uncached,qian2024combining}. Tools like auto-tuners are increasingly being used to help improve application performance by finding optimal application configurations (i.e., optimal write-back and flush behaviors)~\cite{freischuetz2025tuna}. These tools therefore can (and should) be used to guide applications towards the optimal region---which we prescribe is at an \FSYNC{} period $>1000$.}

\begin{tcolorbox}[colframe=white]
    \textbf{Takeaway:} When properly configured, \tool{} delivers 85\% of the throughput of the AEAD baseline (\textbf{\ref{rq:fsync}}). This shows that it is possible to provide strong storage integrity guarantees with minimal overhead.
\end{tcolorbox}

\subsubsection{System Scalability}
Our above analysis showed that if flushes are carefully managed, \tool~can deliver up to a 2.8$\times$ speedup over DMTs and achieve near-zero overhead. Next we examine the scalability of \tool---i.e., whether it can deliver a stable performance guarantee across other system and workload settings that characterize storage deployments.
This will elicit how viable \tool~is as a general solution.

\begin{figure}[t]
    \centering
    \includegraphics[width=\linewidth]{figures/legend.pdf}
    \includegraphics[width=\linewidth]{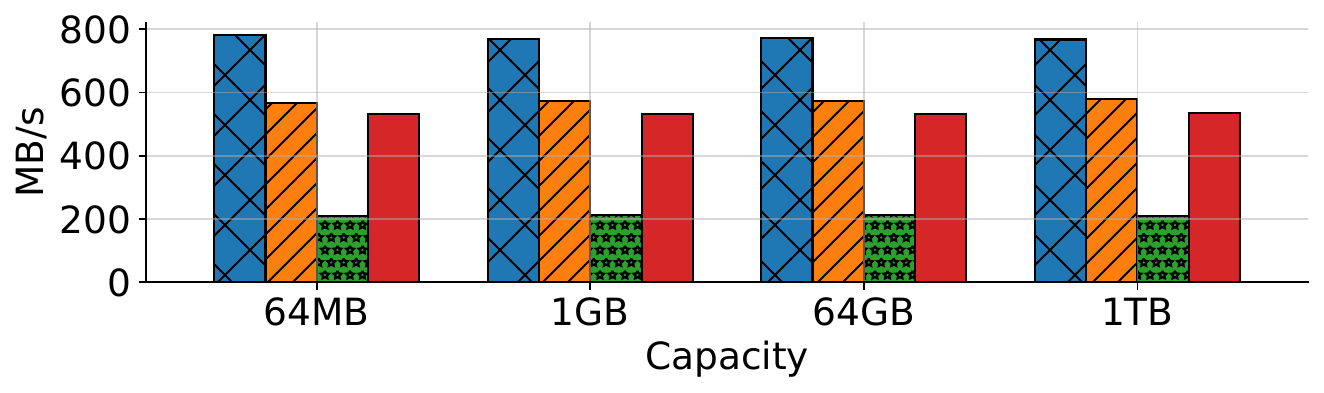}
    \caption{Aggregate read/write throughput vs. capacity. Larger capacities have larger tree heights and thus higher hashing costs. \tool~still delivers near-baseline performance.}
    \label{fig:tp_vs_cap}
\end{figure}

\begin{figure}[t]
    \centering
    \includegraphics[width=\linewidth]{figures/legend.pdf}
    \begin{subfigure}{\linewidth}
        \centering
        \includegraphics[width=\linewidth]{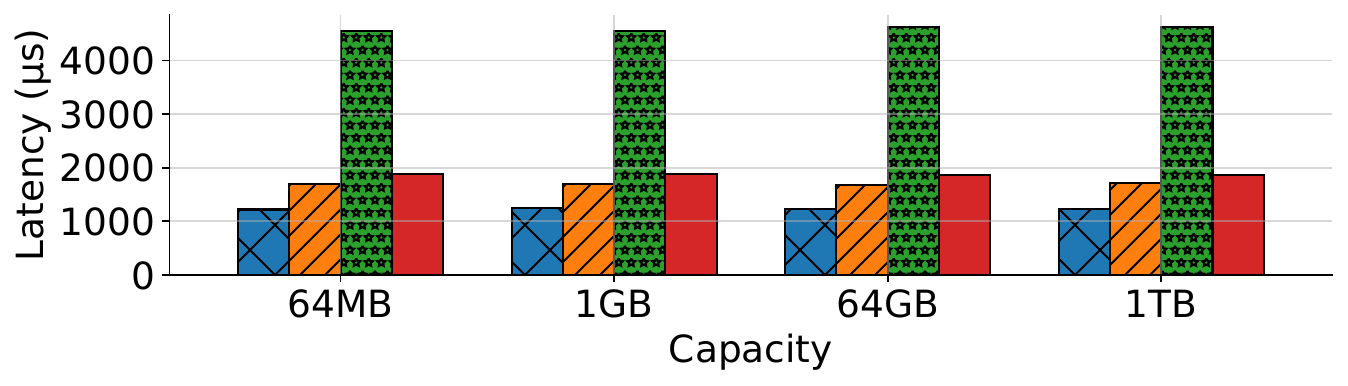}
        \caption{P50 write latency vs. capacity.}
        \label{fig:agg_bw_read_ratio_p50}
    \end{subfigure}
    \hfill
    \begin{subfigure}{\linewidth}
        \centering
        \includegraphics[width=\linewidth]{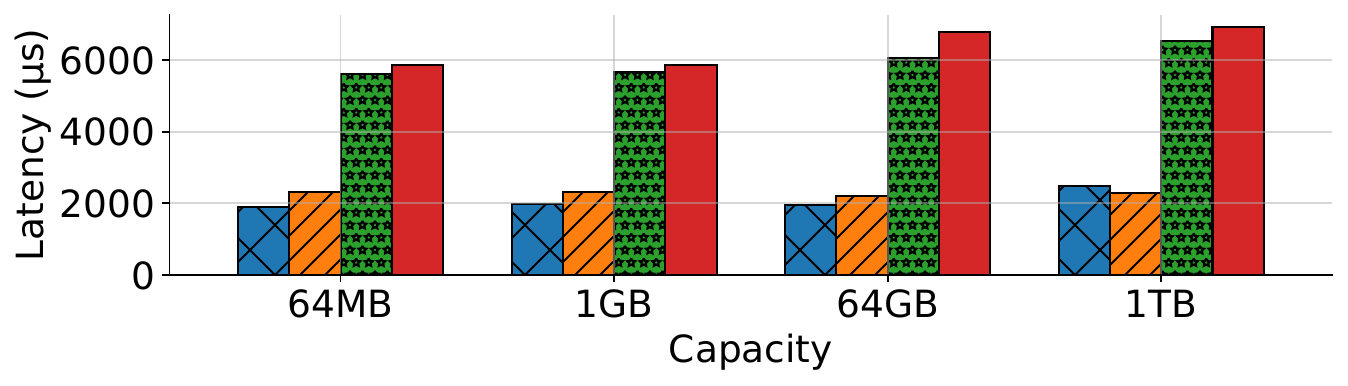}
        \caption{P99.9 write latency vs. capacity.}
        \label{fig:agg_bw_read_ratio_p999}
    \end{subfigure}
    \caption{\tool~median and tail latency improvements reflect throughput improvements.}
    \label{fig:lat_vs_cap}
\end{figure}

\shortsection{Scaling with Capacity}
\autoref{fig:tp_vs_cap} shows how disk capacity impacts performance. Disk capacity affects the size of the tree---notably, the height. Update and verification costs increase at larger capacities because the tree depth increases. We observe across all capacities that aggregate read/write throughput is approximately 200~MB/s for DMTs and 500~MB/s for \tool. This amounts to a 2.5$\times$ speedup and near-baseline performance. Note that DMTs were previously shown to deliver a stable performance guarantee w.r.t. capacity. This shows that \tool~can maintain a stable performance guarantee. Latency improvements reflect these observations. \autoref{fig:lat_vs_cap} shows that median latency has a 2.5$\times$ improvement across all capacities.

\begin{figure}[t]
    \centering
    \includegraphics[width=\linewidth]{figures/legend.pdf}
    \includegraphics[width=\linewidth]{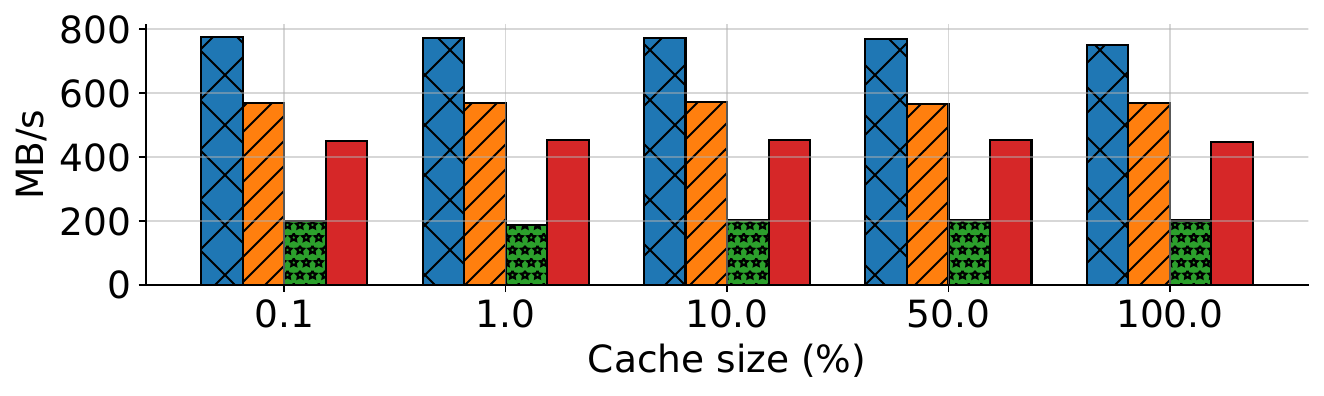}
    \caption{Aggregate read/write throughput vs. cache size. Larger caches help reduce metadata I/O costs.}
    \label{fig:tp_vs_cache}
\end{figure}

\begin{figure}[!t]
    \centering
    \includegraphics[width=\linewidth]{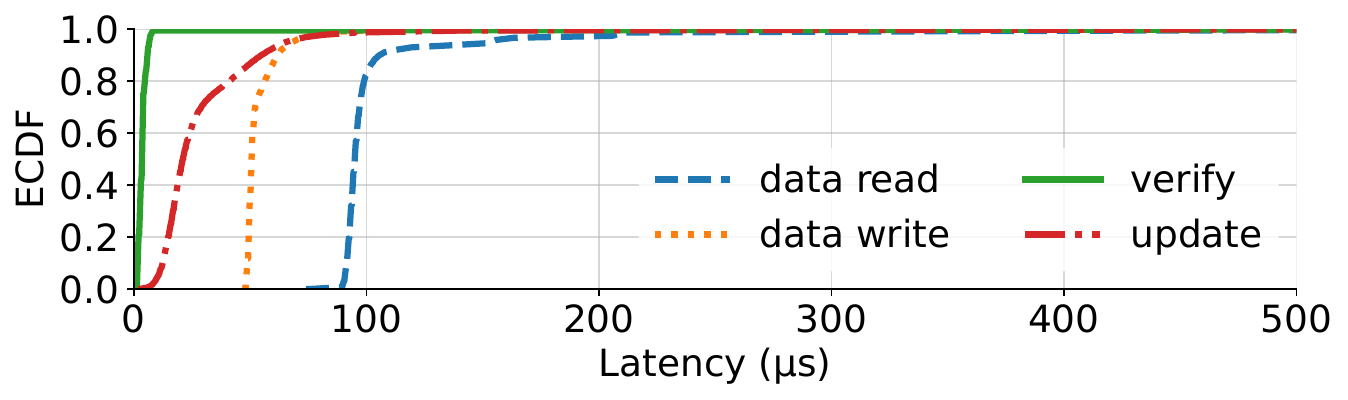}\\
    \caption{Distribution of verification and update latencies. Verification costs are negligible. Thus, deferred verification~\cite{arasu2021fastver,arasu2017concerto} is not only insecure, but it is also unnecessary.}
    \label{fig:verifylat}
\end{figure}

\shortsection{Impact of Cache Size}
Next we examine how the hash cache affects performance. The hash cache plays a central role in minimizing metadata I/O and prompting early exits during verifications. The goal here is to understand whether the limitations of DMTs described in~\autoref{sec:motivation} cannot be solved by simply using a larger cache. 

\autoref{fig:tp_vs_cache} shows that at a cache size of 0.1\%, \tool~delivers a 2.5$\times$ speedup over DMTs and near-baseline performance. In fact, it maintains that speedup to a cache size of 100\%---small caches (e.g., 0.1\%) are very efficient, and performance gains are negligible beyond that.
This cache efficiency is also reflected in the per-block verification and update latencies; for 32~KB I/Os this means that 8 verifications or updates are executed per read/write. \autoref{fig:verifylat} shows the latency distribution at a cache size of 10\%. While updates benefit from caching to an extent, the key insight is that cache hits prompt early exits on nearly all verification operations, making verification latencies largely negligible. We observe that verifications take approximately 2~$\mu$s per block on average, compared to updates which take a median of 25~$\mu$s per block. This is significant because prior works~\cite{arasu2021fastver,arasu2017concerto} have relied on deferring verifications to offset costs. This shows that not only is deferred verification insecure, but it is also unnecessary.

Note that for 72~B nodes (in DMT) and a 1~TB disk, a cache size of 0.1\% amounts to 38.6~MB of memory. Examining cache sizes smaller than this is moot; the effective price for DRAM at these capacities is the same. Thus, the limitations of DMTs cannot be solved by simply using a larger hash cache: efficiently executing updates requires more judicious algorithm design. \tool~provides a promising direction towards this.

\begin{figure}[t]
    \centering
    \includegraphics[width=\linewidth]{figures/legend.pdf}
    \includegraphics[width=\linewidth]{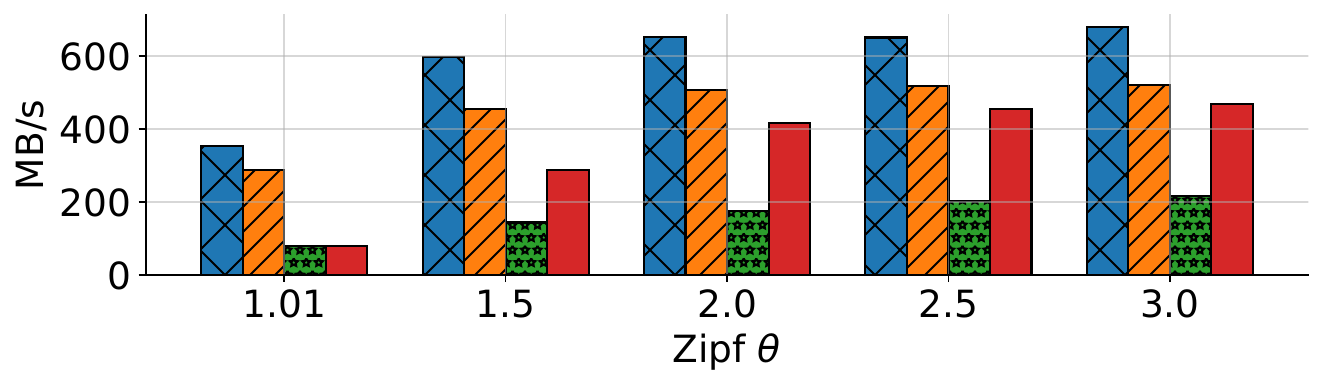}
    \caption{Aggregate read/write throughput vs. workload skewness. More skewed workloads enable \tool~to override updates and better amortize update costs.}
    \label{fig:tp_vs_skew}
\end{figure}

\begin{figure*}[t]
    \centering
    \includegraphics[width=0.5\linewidth]{figures/legend.pdf}\\
    \begin{subfigure}{0.49\linewidth}
        \centering
        \includegraphics[width=\linewidth]{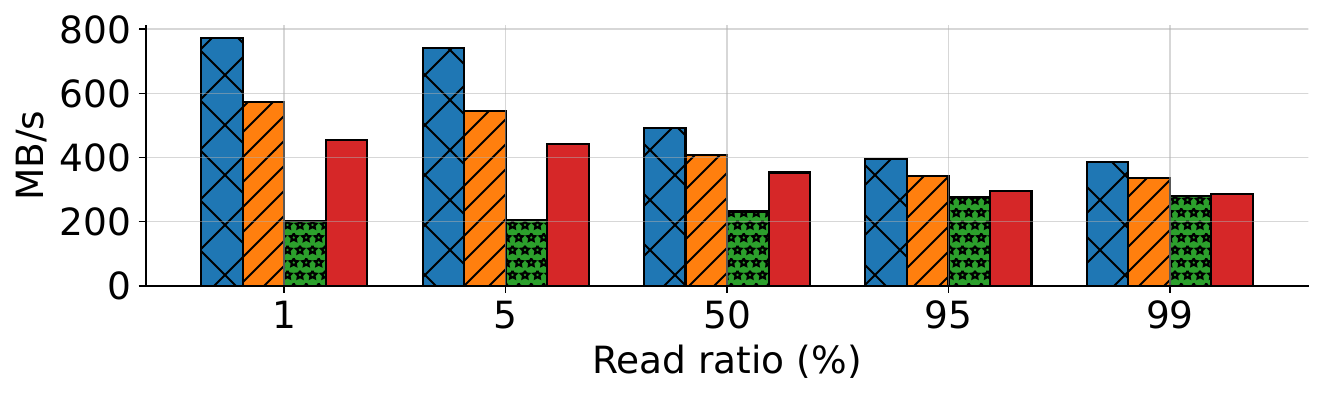}
        \caption{Aggregate read/write throughput vs. read ratio.}
        \label{fig:agg_bw_read_ratio}
    \end{subfigure}
    \hfill
    \begin{subfigure}{0.49\linewidth}
        \centering
        \includegraphics[width=\linewidth]{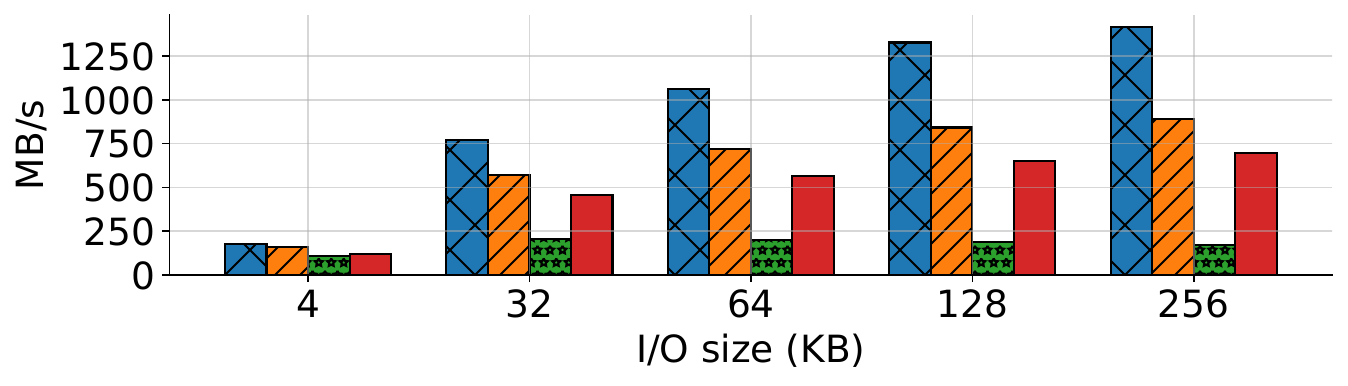}
        \caption{Aggregate read/write throughput vs. I/O size.}
        \label{fig:agg_bw_io_size}
    \end{subfigure}
    \hfill
    \begin{subfigure}{0.49\linewidth}
        \centering
        \includegraphics[width=\linewidth]{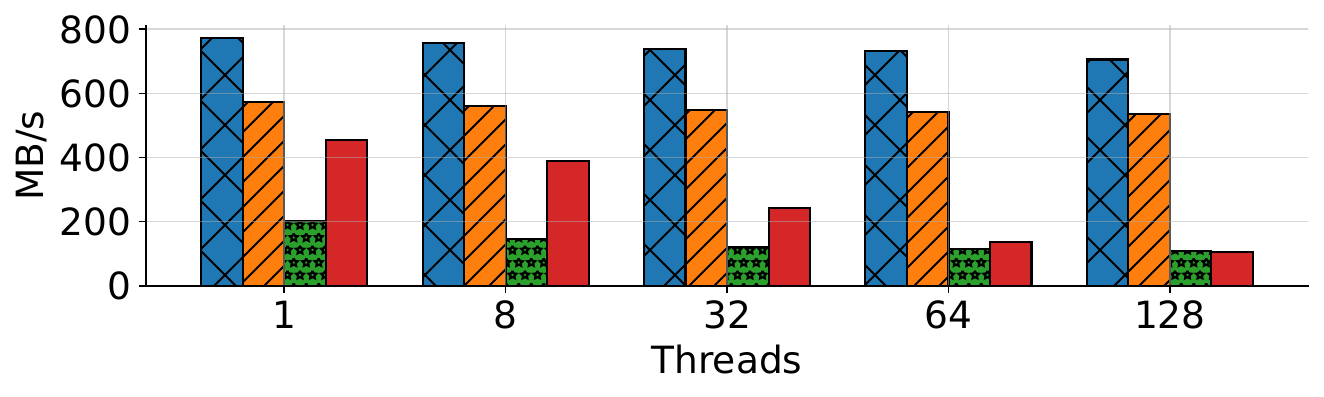}
        \caption{Aggregate read/write throughput vs. thread count.}
        \label{fig:agg_bw_threads}
    \end{subfigure}
    \hfill
    \begin{subfigure}{0.49\linewidth}
        \centering
        \includegraphics[width=\linewidth]{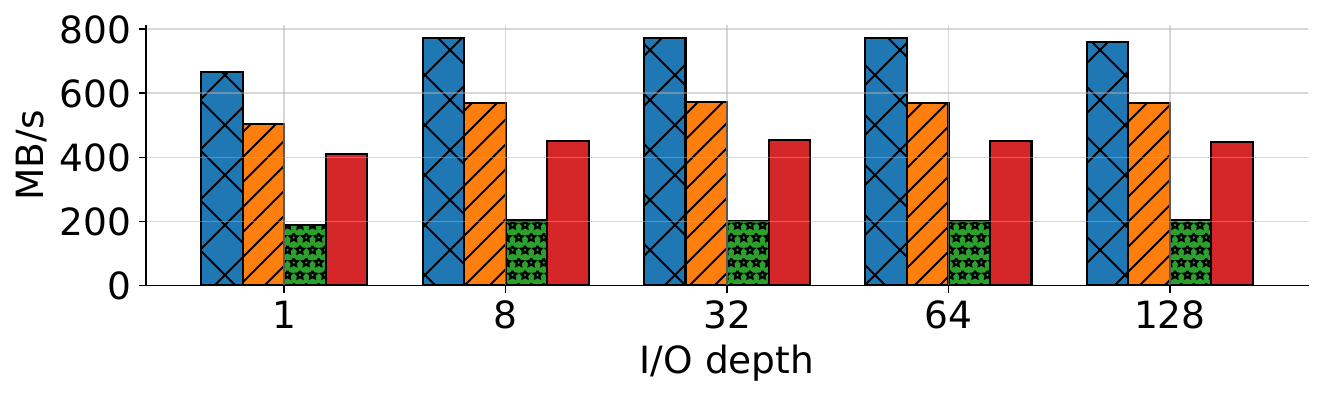}
        \caption{Aggregate read/write throughput vs. I/O depth.}
        \label{fig:agg_bw_depth}
    \end{subfigure}
    \caption{\tool~delivers higher throughputs under write-heavy workloads, larger I/O sizes, and higher I/O depths, and comparable performance otherwise.}
    \label{fig:tp_vs_x}
\end{figure*}

\shortsection{Impact of Workload Skewness}
Next we examine how workload skewness affects performance. The dynamic nature of DMTs enables them to exploit reference locality in skewed workloads, but \tool~can further exploit locality by overriding updates. We therefore anticipate significant speedups under more skewed workloads.

\autoref{fig:tp_vs_skew} shows that \tool~delivers a 2$\times$ speedup over DMTs under a Zipf(1.5) workload and a 2.5$\times$ speedup under a Zipf(2.5) workload. Performance saturates for both trees when $\theta=2.5$. However, \tool~delivers $>$90\% of the AEAD baseline throughput when $\theta>1.5$.

\shortsection{Impact of Read Ratio, I/O Size, Thread Count, and I/O Depth}
Next we examine how other workload parameters affect performance. The read ratio describes the percentage of I/Os that are reads vs. writes; \tool~performance should converge to that of its base data structure (DMT) when workloads are read-heavy, but capitalize on asynchronous updates otherwise. Examining I/O sizes and depths allows understanding how performance changes depending on the nature of how I/Os are issued to the driver from the kernel (e.g., during dirty page writeback). Examining thread counts will elicit whether multiple application threads help or hurt performance in general. We anticipate that \tool~can best amortize update costs under write-heavy workloads, larger I/O sizes, and higher I/O depths. All of these settings provide sufficient conditions to keep the background thread busy while the main thread proceeds with other useful work.

\autoref{fig:agg_bw_read_ratio} shows that \tool~delivers a 2.5$\times$ speedup under write-heavy workloads (read ratio: 1\%). When workloads are read-heavy, \tool~performance is the same as DMTs, because verifications are done synchronously. \tool~scales better with I/O size: at 256~KB I/Os, \tool~observes a speedup of 5.5$\times$. Across application thread counts (\autoref{fig:agg_bw_threads}), a single thread is sufficient to saturate disk throughput; beyond that, contention leads to lower throughputs in general. Across application I/O depths (for applications that use asynchronous I/O on files), \tool~speedups remain largely unaffected.

\begin{tcolorbox}[colframe=white]
    \textbf{Takeaway:} Across storage sizes, workload characteristics, and cache designs, \tool{} provides stable performance guarantees (and up to a 5.5$\times$ speedup over DMTs) that approach those of AEAD (\textbf{\ref{rq:size}}). Even in the worst case, \tool{}'s tail latency is comparable to prior work that provides the same guarantees.
\end{tcolorbox}

\subsubsection{Memory \& Storage Overhead}
\tool~requires additional memory capacity for the queue, but no additional storage capacity over DMTs. DMTs have a 72~B Merkle tree node object. In \tool, request objects in the queue thus occupy $72+16+12=100$~B, for a memory overhead of $1.38\times$. However, we showed that \tool~maintains a $>1.38\times$ speedup with $1.38\times$ less cache memory on average across cache sizes ranging from 0.1\% to 100\%.

\begin{tcolorbox}[colframe=white]
    \textbf{Takeaway:} \tool~delivers better performance per dollar spent on cache memory (\textbf{\ref{rq:storage}}). As a result, it makes efficient use of system resources while providing complete integrity guarantees.
\end{tcolorbox}

\begin{tcolorbox}[colframe=white]
\textbf{Evaluation Summary:}
\tool~shows that, in practical settings, it is possible to achieve near-zero overhead integrity protection. In contrast to conclusions made in prior works, reads never need to be speculative, and challenges associated with writes (state update costs/rollback guarantees) can be largely addressed with judicious parameter selection.
\end{tcolorbox}

\section{Discussion \& Future Work}
\label{sec:discussion}

\shortsection{Optimized Checkpointing}
State updates should in fact be committed before acknowledging to the caller that a state update is durable. However, because sealing costs may be high for some use cases, periodic or batched state updates could be used to improve performance---i.e., only sealing after $N$ state updates. Effectively all prior works have been implemented and evaluated using these methods~\cite{niu2022narrator,angel2023nimble}.

This weakens rollback guarantees. For example, applications rely on flush commands for durability guarantees. Consider that an application receives acknowledgment that a state update $S_0\rightarrow S_1$ is durable, then externalizes some data or acknowledgment to an end-user. If the state update was not actually sealed yet, a malicious crash could be used to roll the system back to state $S_0$ while the user perceives the system to be in state $S_1$. Thus, users would continue execution under the guise of a successful recovery. 

\tool~does not use these optimizations, and we described in our evaluation how \newtxt{practical performance can feasibly be achieved by tuning application write-back parameters (which is standard practice) rather than security parameters.}

\shortsection{\newtxt{Crash Resilience and Recovery}}
\newtxt{Being able to recover from both malicious and non-malicious crashes is an important problem. Currently, \tool{} operates in a fail-stop model, similar to Nimble~\cite{angel2023nimble} and other recent works. This means that correctness is upheld in \tool{} after a crash, but recovery is not guaranteed. How to comprehensively recover from crashes in this context is an open problem that introduces new threat models, security-performance trade-offs, and requires different solutions.}

\newtxt{In general, crashes can be non-malicious or malicious, and they can result from either power loss or software fault (e.g., kernel panic). Non-malicious and malicious crashes resulting from power loss can be handled gracefully by using an uninterruptible power supply or other battery-backed hardware to provide temporary power. Non-malicious and malicious crashes resulting from software faults can be handled (partially) through techniques like journaling. However, an attacker could cause additional crashes during recovery or corrupt, reorder, replay, or truncate parts of the journal. Thus, authentication and freshness of the journal itself must be ensured, among other concerns. This introduces a trade-off space between integrity and availability guarantees that must be modelled and evaluated. We left analysis of these problems to future work.}

\shortsection{Implications for Practical Deployments}
To date, storage integrity methods have either deferred the freshness and consistency problems (e.g., only encrypting and authenticating at the block level), or targeted constrained settings such as key-value stores (where performance requirements are often easier to achieve due to larger object sizes). In contrast, our work demonstrates that disk authenticity, transactional consistency, and freshness can be achieved for general-purpose block devices, and with minimal overhead compared to encryption-only approaches. In the same way that trusted execution environments (TEEs) have enabled general-purpose compute workloads to be directly provided with new guarantees, \tool{} can be used to ensure storage integrity for arbitrary application workloads with untrusted storage hardware.

Such an abstraction enables new practical applications. For instance, abstraction to the block layer enables commodity databases to provide strong storage integrity guarantees. Virtual machine images could be freshness-guaranteed at boot time to prevent downgrade attacks; such protections are increasingly valuable as diverse software supply chains increase the risk of severe security vulnerabilities in older versions of software. Most importantly, these assurances can be achieved without trusting the underlying storage provider. They could prove especially useful for providers that are required to assure protection of user data, such as under the EU GDPR, as both private compute and data storage can now be fully attested even if the underlying physical provider falls outside national jurisdiction.

\shortsection{Asynchronous Integrity Checking as a Pipeline Problem}
Ultimately, faulty data from untrusted storage only has negative security consequences when the application performs operations based on that data. Our current threat model assumes that any code executed based on erroneous data constitutes an integrity violation, but future work could see further performance improvements by precisely characterizing the flow from erroneous data to externally-visible actions. 

In effect, read verification could be treated the same way as memory reads in an out-of-order CPU pipeline. Verification could be deferred until an externally-visible action (e.g., another disk write, network requests, or even mutation of shared memory) is taken as a result of that data. Because read verification \textit{always} succeeds absent an injected fault (an unrecoverable event), application code can proceed under the assumption that verification succeeds, and a segmentation fault can be thrown in the event verification fails. The effectiveness of such an approach will rely on whether side effects of reads can be accurately identified with minimal performance overhead.

\section{Related Work}
\label{sec:related}

Merkle hash trees have emerged as a fundamental building block of secure storage systems, secure memory systems, and authenticated data structures broadly~\cite{buterin2016ethereum,naor2000certificate,android-dm-verity,taassori2018vault,avanzi2022cryptographic,gassend2003caches,mckeen_innovative_2013,erway2015dynamic,arasu2021fastver,sinha2018veritasdb,li2006dynamic,feng2021scalable}. We discuss these related works below.

\shortsection{Secure Storage \& Memory}
Attacks against cloud services have motivated a significant amount of research on building secure storage and memory systems. For storage, authenticated disk encryption has become the standard method to protect the confidentiality and authenticity of cloud disks~\cite{khati2017full,brovz2018practical}. Linux \texttt{dm-verity} implements a block-level Merkle hash tree and is seeing increasing use in emerging cloud services~\cite{aws_bottlerocket}. It has also become critical to providing verified boot for mobile and embedded device storage~\cite{android-dm-verity}. Secure memories like Intel SGX also rely on Merkle trees to protect the integrity (notably, the freshness) of volatile memory~\cite{yan2006improving,gassend2003caches,feng2021scalable,rogers2007using,mckeen_innovative_2013,tsai_graphene-sgx_nodate,priebe_sgx-lkl_2020}.

Recent works such as FastVer~\cite{arasu2021fastver}, Concerto~\cite{arasu2017concerto}, dm-x~\cite{chakraborti2017dm}, and DMT~\cite{burke2025scalable} have designed optimizations to reduce Merkle tree overheads for storage. For example, FastVer and Concerto use the batching approach discussed in~\autoref{sec:motivation}, which sacrifices integrity guarantees. dm-x relies on high-degree (e.g., 128-ary) trees to reduce the tree height and thus the number of hashes that have to be computed on every read or write to storage. DMT implements a hash tree that dynamically adjusts itself at runtime based on workload patterns; this work also showed that the high-degree trees used by dm-x have significant overheads at large capacities and thus high-degree hash trees are not a suitable solution for real-world systems. Recent works such as VAULT~\cite{taassori2018vault} and Penglai~\cite{feng2021scalable} have also demonstrated that hash tree overheads are significant for secure memories and designed optimizations addressing them. These works have relied primarily on using high-degree trees and various caching techniques to reduce overheads. 

Our work builds on these efforts by moving beyond data structure-level optimizations and examining the potential to exploit concurrency to cut costs entirely. We show that with careful design, it is indeed possible to deliver integrity protection with near-zero overhead.

\shortsection{Integrity Data Structures}
The theoretical properties of Merkle hash trees and other authenticated data structures (e.g., authenticated skip lists) have also been examined by the cryptography community in the context of blockchains~\cite{buterin2016ethereum}, certificate revocation systems~\cite{melara2015coniks}, and provable data possession schemes~\cite{erway2015dynamic,dahlberg2016efficient,tamassia2003authenticated,miller2014authenticated,crosby2011authenticated}. These works have broadly examined the problem of integrity checking outsourced storage in an \textit{offline} model. In an offline model, integrity checks occur on-demand as requested by the data owner. This model reflects a batching approach. It suits archival cloud storage, where users issue infrequent integrity checks over their data and performance is not a primary concern.
In contrast, integrity checks are tightly coupled with applications in an \textit{online} model: they are executed each time data is read from or written to disk. This model is necessary for cloud applications that make real-time control decisions based on data read from an untrusted disk, and where performance is a primary concern. 

Prior works have converged on Merkle hash trees as the dominant solution to protecting integrity~\cite{crosby2011authenticated}. However, while their logarithmic complexity has translated into relatively lower overheads for an offline setting, recent works have shown that their overheads in an online setting are often prohibitive~\cite{arasu2021fastver,burke2025scalable,taassori2018vault}. \tool~showed that it is possible to minimize these overheads in the context of high-performance storage devices.

\shortsection{Rollback Protection}
Rollback attacks are similar to but distinct from replay attacks~\cite{niu2022narrator,strackx2016ariadne,parno2011memoir}. Storage systems require protection against both. Our focus is not on designing new rollback protection protocols; secure methods to prevent rollbacks are known---for example, by sealing the Merkle root with a tamper-resistant counter. We discuss rollbacks, because introducing asynchrony complicates rollback guarantees. Specifically, \tool~uses a checkpointing mechanism to ensure the liveness of the background thread (i.e., that state updates are not being delayed) and to ensure that state updates are in fact being committed. This happens during an \texttt{FSYNC} system call.

These prior works do not prescribe precisely \textit{when} state updates should occur, only the means through which an update will occur if requested~\cite{niu2022narrator,strackx2016ariadne,parno2011memoir,brandenburger2017rollback}. The common wisdom is that deciding when it is appropriate to a commit a state update is application-specific and has thus largely been left up to application developers. Optimizations such as batch committing a state change after $N$ state updates have also been proposed to reduce costs associated with sealing. Effectively all prior works have been implemented and evaluated using these methods. The result is a lack of sound guidance on when and how to properly commit state updates, and weak rollback guarantees in general. 

\tool~prescribes that state updates are committed at the time of a flush, which exactly reflects the disk state that applications perceive to be durable. Moreover, our evaluation showed that batching commits is not necessary to achieve good performance: \tool~commits an update on every flush call and still delivers speedups.

\section{Conclusion}
\label{sec:conclusion}

Merkle hash trees provide robust integrity guarantees for data stored in the cloud. However, they introduce additional compute and I/O costs on the I/O critical path, which degrades performance. We showed that closing the performance gap requires moving beyond data structure-level optimizations to algorithm-level optimizations. We proposed a new integrity checking method called \tool~that exploits concurrency to reduce hashing costs on the critical path while still providing equivalent security guarantees. Our evaluation showed that it is possible to achieve near-zero overhead integrity protection for untrusted cloud storage.


\ifCLASSOPTIONcompsoc
  \section*{Acknowledgments}
\else
  \section*{Acknowledgment}
\fi
We thank the anonymous reviewers and shepherd for their insightful feedback. This work was supported in part by the Semiconductor Research Corporation (SRC) and DARPA.


\bibliographystyle{IEEEtran}
\bibliography{IEEEabrv,bib}

\end{document}